\documentclass[11pt,a4paper]{article}

\usepackage[utf8]{inputenc}
\usepackage[T1]{fontenc}
\usepackage{lmodern}

\usepackage{amsmath,amssymb,amsthm}
\usepackage{mathtools}
\usepackage{bm}

\usepackage[margin=2.5cm]{geometry}
\usepackage{setspace}
\usepackage{parskip}

\usepackage{graphicx}
\graphicspath{{figures/}}

\usepackage{booktabs}
\usepackage{array}
\usepackage{tabularx}
\usepackage[section]{placeins}   %
\usepackage{enumitem}

\usepackage[colorlinks=true,linkcolor=blue,citecolor=blue,urlcolor=blue]{hyperref}
\usepackage[capitalise,noabbrev]{cleveref}

\usepackage[authoryear,round]{natbib}
\newtheorem{theorem}{Theorem}

\theoremstyle{definition}
\newtheorem{definition}{Definition}

\theoremstyle{remark}
\newtheorem*{remark}{Remark}

\newcommand{\rate}{r}                 %
\newcommand{\sched}{\tau}             %
\newcommand{\tw}{\tau_w}              %
\newcommand{\dd}{\mathrm{d}}
\newcommand{\RR}{\mathbb{R}}

\newcommand{\EE}{\mathbb{E}}

\newcommand{\Pcal}{\mathcal{P}}
\newcommand{\Wtwo}{W_2}
\newcommand{\Neut}{\mathcal{N}}
\newcommand{\Schd}{\mathcal{S}}

\newcommand{\Free}{F}                 %
\newcommand{\dF}{\Delta F}            %

\newcommand{\one}{\mathbf{1}}
\DeclareMathOperator*{\argmin}{arg\,min}

\title{Minimum-Distortion Wealth Taxation, I:\\
Information-Theoretic versus Transport-Geometric\\
Optimality on the Proportional Class}
\author{Anders G.\ Fr{\o}seth\thanks{Independent Researcher.\
  E-mail: \href{mailto:indrefjorden@pm.me}{indrefjorden@pm.me}.}}
\date{\today}

\begin{document}
\maketitle

\begin{abstract}
\noindent
We characterise minimum-distortion wealth taxation under two
contrasting normative criteria within a Fokker--Planck framework
on log-wealth: the JKO free-energy gap, an
information-theoretic measure aligned with the Mirrleesian
decision-distortion tradition, and the squared 2-Wasserstein
distance from the no-tax distribution at horizon $T$, a
transport-geometric measure aligned with the Saez--Zucman
distributional-compression tradition. Restricting to the
neutrality-preserving (C1)--(C3) schedule class of
\citet{Froeseth2026F}, both optima admit closed forms in the
two-dimensional design plane parametrised by the
corporate--dividend retention $k = (1-\tau_c)(1-\tau_d)$ and the
proportional wealth-tax rate $\tw$. The JKO optimum partitions
the regime axis into three phases as a function of the
dimensionless ratio $\rho = \Sigma_0\,m_0/\sigma^2$, with
$m_0 = \mu - \sigma^2/2$ the geometric mean log-return: a pure
wealth-tax phase at low $\rho$, a mixed-instrument phase at
intermediate $\rho$, and a pure flow-tax phase at high $\rho$.
The $W_2$ optimum, by contrast, is degenerate in this
calibration: it pins to the pure flow-tax corner across the
whole regime axis. The criterion contrast admits an
economically meaningful reading via a \emph{bluntness index}
$B(m_0) = b/(a m_0)$ that measures the wealth-tax channel's
mean-displacement-per-revenue overshoot relative to the flow-tax
channel; JKO weights $B$ linearly, $W_2$ weights it
quadratically, and the two normative traditions correspond to
this difference in weighting. Norwegian-flavoured calibrations
sit inside the JKO mixed-instrument phase under stock-heavy
portfolio volatility but move into the pure flow-tax phase under
the realised effective volatility of typical real-estate-heavy
households. We do not in this paper compare directly to the
bracket-based wealth-tax schedules in actual use; that
comparison requires a piecewise-Gaussian extension of the
present analysis, sketched as a companion paper.
\end{abstract}

\section{Introduction}\label{sec:intro}

\subsection{The question}\label{sec:question}

The political debate around wealth taxation has been arrayed along a
familiar normative axis for decades, with four arguments structuring
the standard exchange. Proponents typically advance two claims.
\emph{First}, that capital owners accumulate disproportionate
influence and political power relative to wage earners, and that
taxation of the capital \emph{stock} --- not merely the income flow
it generates --- is the appropriate instrument through which the
wealthier should contribute a larger share to public revenue
\citep{Piketty2014, PikettyZucman2014}.
\emph{Second}, that a wealth tax is the most direct and effective
instrument for distributional redistribution: where income taxes
track flows that capital owners can defer or restructure, a wealth
tax acts on the stock itself, and is therefore the central
redistributive tool in any serious progressive programme
\citep{SaezZucman2019, BlanchetFournierPiketty2022}.
Opponents typically advance two countervailing claims.
\emph{Third}, that a wealth tax punishes entrepreneurship and
discourages the deployment of risky capital, because it taxes the
position rather than the realised return and falls heaviest on
risk-bearing assets at exactly the points in the wealth lifecycle
where allocation decisions are most consequential \citep{Mirrlees2011,
SaezStantcheva2018}. \emph{Fourth}, that a wealth tax distorts
capital allocation and disadvantages domestic investors relative to
foreign investors who do not face the same levy in their home
jurisdictions, with empirical responses documented in
\citet{JakobsenEtAl2020} and the Norwegian B{\o} reform debate.

These four arguments share a common structure once translated into
the language of distortion-minimisation. The two pro-tax arguments
say, in effect, that a wealth tax should be designed to minimise the
\emph{distributional} perturbation from the no-tax counterfactual
--- to compress the upper tail, or to move households toward greater
equality at finite horizon. The two anti-tax arguments say, in
effect, that a wealth tax should be designed to minimise the
\emph{allocational} perturbation --- not to redirect the household's
portfolio decisions, and not to favour foreign over domestic
capital. Both normative positions are defensible. Neither is
entailed by the other, and as we will show, they correspond to two
mathematically distinct distortion criteria. The
distributional-compression position aligns with the
2-Wasserstein distance between post-tax and no-tax distributions
($\Wtwo^2$) --- the natural geometric metric on probability
measures --- and with the Saez--Zucman tradition in optimal-tax
theory. The allocational-neutrality position aligns with the
JKO free-energy gap --- the natural information-theoretic
distortion in the Fokker--Planck framework --- and with the
Mirrleesian decision-distortion tradition. Section~\ref{sec:why-these-two}
develops the criterion choice in detail; Section~\ref{sec:traditions}
develops the connection to the two normative traditions.
The political debate is then a debate between two implicit
objective functions, both well-defined, and the choice between
them is a normative one that the present paper does not adjudicate.

What the paper does is make those two objectives precise,
characterise the wealth-tax schedule that minimises each, and prove
that the two optima are governed by a single dimensionless ratio
that distinguishes regimes in which they coincide from regimes in
which they diverge.

\subsection{Two distortion criteria}\label{sec:two-criteria}

Working within the Fokker--Planck framework of \citet{Froeseth2026S,
Froeseth2026F} on log-wealth, we restrict attention to the
\emph{(C1)--(C3) class} of \citet{Froeseth2026F} --- tax schedules
satisfying three conditions on the underlying ownership-tax system:
the capital-income tax rate matches the corporate tax rate
($\tau_k = \tau_c$); the shielding rate matches the risk-free rate
($r_s = r_f$); and the wealth-tax assessment is uniform across
asset classes ($\alpha_i = \alpha$). Table~\ref{tab:c1c3} of
Section~\ref{sec:c1c3} states each condition formally with its
effect.
Together the three conditions imply that the post-tax wealth
process inherits the multiplicative structure of the pre-tax one,
with the effect on log-wealth at finite horizon being a
\emph{drift-shift-and-rescale}: a uniform shift from the wealth
tax and a uniform rescaling of excess drifts by the factor
$k = (1 - \tau_c)(1 - \tau_d)$ from the corporate--dividend
flow-tax pair. The design space is parametrised by this pass-through
factor $k \in (0, 1]$ together with the proportional wealth-tax
rate $\tw \in [0, \tw^{\max}]$. By
\citet{Froeseth2026N, Froeseth2026E, Froeseth2026F} this is exactly
the schedule class that preserves \emph{neutrality} with respect to
portfolio choice under homogeneous returns and CRRA preferences:
the household's risk-bearing decisions are not redirected by a tax
in this class. Schedules outside (C1)--(C3) break neutrality.

\paragraph{The design parameters are direct policy levers.}
The substantive importance of this decomposition cannot be
overstated: the design parameters $(k, \tw)$ are not abstract
optimisation variables, but the actual policy levers in the
ownership-tax debate. The wealth-tax rate $\tw$ is the
proportional rate the public discussion centres on. The pass-through
factor $k = (1-\tau_c)(1-\tau_d)$ is the share of pre-tax returns
that survives the corporate--dividend tax sequence: $k = 1$
corresponds to no flow taxation, $k = 0.5$ corresponds to a
combined corporate--dividend burden of about half pre-tax
returns. Choosing a point in $\Schd_C = (0, 1] \times [0, \tw^{\max}]$
\emph{is} choosing an actual ownership-tax package satisfying
(C1)--(C3). This is the link between the framework's analytical
content and the political-economy debate of Section~\ref{sec:question}:
arguments about the wealth tax and arguments about
corporate--dividend taxation are arguments about the two axes of
$\Schd_C$.

Within this class, we characterise the schedule that minimises each
of two distortion criteria.

The \emph{Jordan--Kinderlehrer--Otto (JKO) free-energy gap}
\citep{JordanKinderlehrerOtto1998}
\begin{equation}\label{eq:dF-intro}
\dF[\sched] \;:=\; \Free[p_T^{\sched}] - \Free[p_T^0],
\qquad
\Free[p] \;=\; \int p \log p \,\dd y + \int V_Y^0\, p \,\dd y,
\end{equation}
is the information-theoretic distortion measure native to the
Fokker--Planck framework. It measures how far the post-tax population
distribution at horizon $T$ has been pushed from the no-tax counterfactual
in entropy-plus-potential terms. The JKO criterion is the natural
formalisation of \emph{minimum decision-distortion}: the policy that
minimises $\dF$ at matched revenue is the policy that least disturbs
the natural Fokker--Planck flow of household wealth, and therefore
the one that least redirects the household's allocation and
risk-bearing decisions. The two anti-tax arguments above are
implicit JKO concerns.

The \emph{squared 2-Wasserstein distance} to the no-tax distribution
\begin{equation}\label{eq:W2-intro}
\Wtwo^2(p_T^{\sched}, p_T^0)
\;=\; \int_0^1 \big[F_T^{\sched,-1}(u) - F_T^{0,-1}(u)\big]^2 \dd u,
\end{equation}
introduced in this context by \citet{Froeseth2026W}, is the
transport-geometric distortion measure: it measures how much
``mass'' must be moved from the no-tax distribution to reach the
post-tax distribution, percentile by percentile. The W$_2$ criterion is
the natural formalisation of \emph{minimum distributional
perturbation}: the policy that minimises $\Wtwo^2$ at matched
revenue is the policy that keeps every percentile of the wealth
distribution closest to its no-tax counterpart. The two pro-tax
arguments above are implicit W$_2$ concerns when restricted to
revenue-matched schedules.

These two criteria are not the same. JKO weights perturbations to
the distribution linearly in both the mean shift and the (relative)
variance shift, while W$_2$ weights both quadratically. The
quantitative consequence is that the criteria pick different optima
within the same (C1)--(C3) design space.

\subsection{Results}\label{sec:results-intro}

We establish four results. \emph{First} (Theorems~\ref{thm:wellposed} and~\ref{thm:jkoopt}), the constrained minimisation of each criterion within the (C1)--(C3) class admits a unique closed-form optimum; the JKO optimum has the structure ``linear in $\tw$, log-derivative in $k$,'' and the W$_2$ optimum is quadratic in both. \emph{Second} (Theorem~\ref{thm:crossover}), the JKO optimum partitions the regime axis into three phases as a function of the dimensionless ratio
\begin{equation}\label{eq:rho-intro}
\rho \;:=\; \Sigma_0 \cdot m_0/\sigma^2
\;=\; \sqrt{v_0 + \sigma^2 T}\,\frac{\mu - \sigma^2/2}{\sigma^2}
\end{equation}
combining the horizon variance $\Sigma_0$, the log-wealth drift
$m_0$, and the volatility squared: a pure wealth-tax phase at low
$\rho$, a mixed-instrument phase at intermediate $\rho$, and a
pure flow-tax phase at high $\rho$. The two optima coincide as
$\rho \to \infty$ and diverge as $\rho \to 0$. \emph{Third}, the
phase structure has a phase-transition reading
(Section~\ref{sec:phase-transition}): the wealth-tax revenue share
$\phi^{\star} = a\,\tau_w^{\star}/R^{\star}$ acts as an order
parameter, $\rho$ as a control parameter, and the phase
boundaries $\rho_{\rm low}, \rho_{\rm high}$ as kinks in
$\phi^{\star}(\rho)$. The corresponding $W_2$ analysis on
(C1)--(C3) yields a degenerate diagram with no transition --- a
``phase-richness contrast'' between the two criteria.
\emph{Fourth} (calibration in Section~\ref{sec:calibration}), the
Norwegian-flavoured regime sits at $\rho \approx 0.23$, inside
the mixed-instrument phase under stock-heavy portfolio
volatility, but moves toward the corner phases when effective
volatility is computed from typical real-estate-heavy household
portfolios.

The paper does not adjudicate between the two criteria, and it
does not in this paper attempt a direct comparison between the
JKO recommendation and the actual schedules in force in Norway,
Switzerland, Spain, or France --- those are bracket-based and lie
outside the proportional (C1)--(C3) class on which the closed
forms here are derived. The matching empirical exercise lives on
the bracket sub-class, the subject of a companion paper sketched
in Section~\ref{sec:open-questions}. What this paper does is articulate
the normative split between two distortion criteria as a precise
mathematical split, characterise the optimal proportional
schedule under each, and identify when the two sides converge and
when they diverge across the regime axis.

\subsection{Placement in the corpus}\label{sec:placement}

The present paper is the FP-side companion to \citet{Froeseth2026W}
(``the Wasserstein paper''). \citet{Froeseth2026W} establishes the
W$_2$ optimality of the progressive schedule
$\sched^{\star}(x) = \lambda x^2$ over the full admissible class
$\Schd$, including schedules outside (C1)--(C3) that break
neutrality. The present paper restricts to (C1)--(C3) and asks how
the JKO and W$_2$ criteria compare \emph{within} the
neutrality-preserving subclass. Together the two papers identify a
2$\times$2 grid of design problems: criterion JKO or W$_2$, schedule
class (C1)--(C3) or full $\Schd$. The four cells fill in
characterisations of minimum-distortion taxation under different
combinations of normative criterion and admissible-schedule
restrictions.

\paragraph{Part I of a two-paper series.} The present paper
treats the proportional sub-class --- the (C1)--(C3) instruments
parametrised by $(k, \tw) \in \Schd_C$ --- and develops the
phase structure of the JKO optimum in closed form. It mirrors
\citet{Froeseth2026N} in establishing the theory cleanly on a
restricted class, with the same restriction-to-extension structure
that \citet{Froeseth2026E} bears to \citet{Froeseth2026N}. A
companion paper, in preparation, extends the analysis to the
\emph{2-bracket sub-class} $(k, X_0, r)$, the simplest
finite-dimensional schedule class that captures the structure
actually used by Norway, Switzerland, Spain, and France's old ISF
(an exemption threshold combined with marginal rates above). The
extension is technically tractable via piecewise-Gaussian FP
solutions matched at the bracket boundaries. Section~\ref{sec:open-questions}
sketches the recipe and the relationship to the present paper.
Direct comparisons between the recommendations of the present
analysis and the actual schedules in force in any specific
jurisdiction are deferred to that companion paper, since the
proportional class on which the present results are derived is
the wrong functional space in which to make such comparisons.

The paper also draws on \citet{Froeseth2026R} for the drift-design
language used in the discussion of (C3) violations, and on a
forthcoming empirical pipeline for the calibration in
Section~\ref{sec:calibration}. Heterogeneous-returns extensions
(in the spirit of \citet{Froeseth2026H}) and migration-channel
considerations (\citet{Froeseth2026C}) are noted as orthogonal
extensions in the discussion.

\section{Setup}\label{sec:setup}

The symbols used throughout the paper are summarised in
Table~\ref{tab:notation}; the rest of Section~\ref{sec:setup} introduces them
in context.

\begin{table}[!htbp]
\centering
\small
\begin{tabular}{l p{0.78\textwidth}}
\toprule
\multicolumn{2}{l}{\textbf{Wealth process and parameters}} \\
\midrule
$X_t$ & Pre-tax wealth process; $X_0$ has distribution $\pi_0$. \\
$Y_t = \log X_t$ & Log-wealth process. \\
$\mu, \sigma$ & GBM drift and volatility on raw wealth. \\
$m_0 := \mu - \sigma^2/2$ & Drift on log-wealth under no taxation; equivalently the
                              \emph{expected log-wealth growth rate} (volatility-corrected
                              geometric mean return) of pre-tax wealth. \\
$T$ & Finite-horizon design parameter. \\
$\pi_0 \in \Pcal_2(\RR_+)$ & Initial wealth distribution. \\
$\mu_0, v_0$ & Mean and variance of $\log X_0$ when $X_0$ is log-normal. \\
\midrule
\multicolumn{2}{l}{\textbf{Population distributions}} \\
\midrule
$p_t^0$ & No-tax distribution of $Y_t$ at time $t$ (Fokker--Planck flow under no tax). \\
$p_t^{\sched}$ & Post-tax distribution of $Y_t$ at time $t$ under schedule $\sched$. \\
$M_t^{\sched}, \Sigma_t^{\sched}$ & Mean and standard deviation of $p_t^{\sched}$ when Gaussian. \\
$\Sigma_0 := \sqrt{v_0 + \sigma^2 T}$ & Standard deviation of $Y_T$ under no tax. \\
\midrule
\multicolumn{2}{l}{\textbf{Tax instrument and schedule classes}} \\
\midrule
$\rate(x)$ & Wealth-tax rate per unit wealth
                  (rate convention of \citet{Froeseth2026N, Froeseth2026F}). \\
$\sched(x) := \rate(x) \cdot x$ & Tax schedule: total tax paid by a holder of wealth $x$
                                  per unit time (schedule convention of
                                  \citet{Froeseth2026W} and the present paper). \\
$\tw$ & Proportional wealth-tax rate (the constant rate in the
        (C1)--(C3) class). \\
$k = (1-\tau_c)(1-\tau_d)$ & Pass-through factor (corporate--dividend flow-tax product),
                              $k \in (0, 1]$. Post-tax log-wealth volatility is $\sqrt{k}\,\sigma$. \\
$\Schd$ & Full admissible schedule class. \\
$\Schd_C := (0, 1] \times [0, \tw^{\max}]$ & (C1)--(C3) design space. \\
$\Schd_R \subset \Schd_C$ & Revenue-matched constraint set ($\bar R = R^{\star}$). \\
$\Neut$ & Neutral schedule class of \citet{Froeseth2026N}. \\
\midrule
\multicolumn{2}{l}{\textbf{Distortion criteria and reference quantities}} \\
\midrule
$V_Y^0(y) := -m_0 y/\sigma^2$ & Reference potential (no-tax FP gradient flow). \\
$\Free[p] = \int p \log p + \int V_Y^0 p$ & JKO free-energy functional. \\
$\dF[\sched] := \Free[p_T^{\sched}] - \Free[p_T^0]$ & Terminal-gap JKO criterion. \\
$\Wtwo^2(p, q)$ & Squared 2-Wasserstein distance between distributions
                    $p$ and $q$. \\
$R, R^{\star}$ & Time-averaged revenue rate; revenue target. \\
$a, b$ & Linear-approximation revenue weights:
        $R \approx a\,\tw + b\,(1-k)$. \\
\midrule
\multicolumn{2}{l}{\textbf{Crossover}} \\
\midrule
$\rho := \Sigma_0\,m_0/\sigma^2$ & Dimensionless crossover ratio
                                    governing the JKO/$W_2$ split
                                    (Section~\ref{sec:crossover}). \\
\bottomrule
\end{tabular}
\caption{Summary of notation. Variables introduced once and used
locally are not listed; see the surrounding text at first
appearance.}
\label{tab:notation}
\end{table}

\subsection{Pre-tax dynamics}

We work throughout in log-wealth coordinates $Y = \log X$. Pre-tax
wealth $X_t$ follows GBM with drift $\mu$ and volatility $\sigma$:
\begin{equation}\label{eq:sde-X}
\dd X_t \;=\; \mu X_t \,\dd t + \sigma X_t \,\dd W_t, \qquad
X_0 \sim \pi_0.
\end{equation}
By It\^o's lemma, $Y$ satisfies
\begin{equation}\label{eq:sde-Y}
\dd Y_t \;=\; m_0 \,\dd t + \sigma\, \dd W_t, \qquad
m_0 := \mu - \sigma^2/2,
\end{equation}
with constant diffusion in log-wealth.

\subsection{Notation: rate versus schedule}\label{sec:notation}

The corpus uses two different conventions for tax-schedule
notation, and reconciling them needs to be explicit.
\citet{Froeseth2026N, Froeseth2026E, Froeseth2026F} use $\tau$
(or $\tw$) for the \emph{wealth-tax rate}: a holder of wealth
$x$ pays $\tw \cdot x$ per unit time. \citet{Froeseth2026W} (the
Wasserstein paper) uses $\tau(x)$ for the \emph{tax schedule}
(total amount paid).

We adopt an explicit two-symbol convention. Let $\rate(x)$ denote
the wealth-tax rate per unit wealth per unit time, consistent with
the $\tw$-as-rate convention of \citet{Froeseth2026N}. Let
$\sched(x) := \rate(x) \cdot x$
denote the tax schedule (total tax paid by a holder of wealth $x$
per unit time), consistent with the variational object of
\citet{Froeseth2026W}. Each policy can be described cleanly via
either object:

\begin{center}
\begin{tabular}{lll}
\toprule
Policy & Rate $\rate(x)$ & Schedule $\sched(x)$ \\
\midrule
Proportional wealth tax & $\tw$ (constant) & $\tw\,x$ \\
W$_2$-optimal of \citet{Froeseth2026W} & $\lambda x$ & $\lambda x^2$ \\
Two-bracket (threshold $x_0$, rate $r_0$) & $r_0\,\one\{x \ge x_0\}$ & $r_0\,(x - x_0)_+$ \\
\bottomrule
\end{tabular}
\end{center}

The drift modification on log-wealth from a stock tax is determined
by the rate: $\delta v_Y(y) = -\rate(e^y)$.

\subsection{Tax instruments and the (C1)--(C3) class}\label{sec:c1c3}

A flow + stock tax instrument $\sched$ lies in the
\emph{(C1)--(C3) class} of \citet{Froeseth2026F} if it satisfies
the three conditions summarised in Table~\ref{tab:c1c3}.

\begin{table}[!htbp]
\centering
\begin{tabular}{c p{0.30\textwidth} p{0.50\textwidth}}
\toprule
\textbf{Condition} & \textbf{Statement} & \textbf{Effect} \\
\midrule
(C1) &
Capital-income tax rate equals corporate tax rate:
$\tau_k = \tau_c$. &
Removes the differential between taxation of dividends paid out
versus capital income retained, so that the after-tax excess
return is independent of how a unit of pre-tax return is realised.
\\[4pt]
(C2) &
Shielding rate equals the risk-free rate: $r_s = r_f$. &
The shielding mechanism (Norwegian \emph{skjermingsfradrag}) acts
as a pure deduction at the risk-free rate, equating after-tax
excess returns across assets up to the wealth-tax channel.
Together with (C1) this ensures the equity--debt split is
undistorted.
\\[4pt]
(C3) &
Uniform wealth-tax assessment across assets:
$\alpha_i = \alpha$ for all assets~$i$. &
No asset-class-specific drift modification; the wealth tax acts
as a uniform drift shift on log-wealth. Violations distort the
tangency portfolio.
\\
\bottomrule
\end{tabular}
\caption{The three conditions of the (C1)--(C3) class, restated
from \citet{Froeseth2026F}. \emph{Together} they imply that the
combined tax system acts on the Fokker--Planck equation as a
uniform drift shift (from the wealth tax) plus a uniform rescaling
of excess drifts by the factor $k = (1-\tau_c)(1-\tau_d)$ (from
the corporate--dividend flow-tax pair). The household's
portfolio-choice problem is unchanged by the tax up to this
rescaling; \citet{Froeseth2026F} calls this property
\emph{generalised neutrality}.}
\label{tab:c1c3}
\end{table}

Specifically, \citet{Froeseth2026F} establishes that under
(C1)--(C3) the post-tax log-wealth process takes the
drift-shift-and-rescale form
\begin{equation}\label{eq:sde-Ytax}
\dd Y_t \;=\; \big[ k\,m_0 - \tw \big]\,\dd t + \sqrt{k}\,\sigma\,\dd W_t,
\end{equation}
where the pass-through factor
$k = (1 - \tau_c)(1 - \tau_d) \in (0, 1]$ is determined by the
corporate and dividend tax rates (with $k = 1$ corresponding to
the no-flow-tax case $\tau_c = \tau_d = 0$, and $k \to 0$
corresponding to full pass-through), and the proportional
wealth-tax rate $\tw \in [0, \tw^{\max}]$ enters as a uniform
drift shift on log-wealth.

\paragraph{The link between policy levers and Fokker--Planck channels.}
Equation~\ref{eq:sde-Ytax} makes precise an idea that motivates everything
that follows. The two design parameters $(k, \tw)$ enter the
Fokker--Planck equation through \emph{distinct channels}: the
wealth-tax rate $\tw$ acts purely as a uniform drift shift, while
the corporate--dividend pass-through $k$ acts purely as a
multiplicative rescaling of both drift and diffusion. The drift-shift
mechanism is the substance of \citet{Froeseth2026N}'s neutrality
result for the proportional wealth tax in the single-asset GBM
setting; the rescaling-by-$k$ mechanism is \citet{Froeseth2026F}'s
generalisation when the full ownership-tax system (corporate +
dividend + capital-income + wealth) is in play. The (C1)--(C3) class
is therefore not a \emph{technical} restriction so much as the
\emph{natural decomposition} of how policy operates on the wealth
process: every ownership-tax package satisfying (C1)--(C3) reduces to
some choice of $(k, \tw)$, and conversely every choice of $(k, \tw)$
in $\Schd_C$ corresponds to an actual implementable ownership-tax
package.

This decomposition is what allows the design problem of this paper
to be stated cleanly. We are not optimising over abstract function
spaces; we are choosing a single point in the two-dimensional
$(k, \tw)$-plane, with each axis a familiar policy lever and the
constraint set $\Schd_R$ a one-dimensional curve corresponding to
matched-revenue trade-offs. The substantive question --- how much of
public revenue should come from the wealth-tax channel ($\tw$) and
how much from the corporate--dividend channel ($k$) --- is exactly the
question the JKO and $W_2$ optima of Section~\ref{sec:variational} answer,
each from its own normative criterion.

We denote the resulting design space by
\begin{equation}\label{eq:design-space}
\Schd_{C} \;:=\; (0, 1] \times [0, \tw^{\max}].
\end{equation}

The (C1)--(C3) class is at once a substantive and a technical
restriction. Substantively, it is the class of \emph{neutrality-
preserving} schedules in the sense of \citet{Froeseth2026N} ---
schedules under which the household's risk-bearing decisions are
not redirected by the tax. Technically, it reduces the
infinite-dimensional schedule space to the two-parameter design
space $\Schd_{C}$ above, in which existence and uniqueness of
optima follow from standard convex-analysis arguments
(Section~\ref{sec:wellposed}). Schedules outside (C1)--(C3) break
neutrality and can be analysed in the broader framework of
\citet{Froeseth2026W} or \citet{Froeseth2026R}; the present paper
restricts to (C1)--(C3) by design.

\subsection{Population measure and Fokker--Planck dynamics}\label{sec:fp}

The population is a Borel probability measure $\pi_0 \in
\Pcal_2(\RR_+)$ on initial wealth, with finite second moments on
log-wealth so that $Y_0 \in \Pcal_2(\RR)$. Specific results that
involve closed-form ratios additionally assume $X_0$ to be
log-normal with log-volatility $\sigma_0$; we flag this assumption
at the point of use rather than treating it as standing.

The population law of log-wealth at time $t \in [0, T]$ under no
taxation is the pushforward $p_t^0 := \mathrm{Law}(Y_t \,|\,
\eqref{eq:sde-Y})$, satisfying the Fokker--Planck equation
\begin{equation}\label{eq:fp-notax}
\partial_t p_t^0 \;=\; -m_0\, \partial_y p_t^0
+ \tfrac{\sigma^2}{2}\, \partial_y^2 p_t^0.
\end{equation}
Under a schedule $\sched \in \Schd_C$ the post-tax log-wealth has
modified drift and diffusion as in \eqref{eq:sde-Ytax}, and its
population law $p_t^{\sched}$ satisfies the FP equation with the
modified coefficients.

The free-energy functional $\Free$ of \eqref{eq:Fdef} is the JKO
gradient flow object for \eqref{eq:fp-notax}: the no-tax FP
equation is the steepest-descent flow of $\Free$ in the
2-Wasserstein metric on $\Pcal_2(\RR)$
\citep{JordanKinderlehrerOtto1998, Santambrogio2016, Villani2009}.
This is the structural reason why $\Free$ is the natural
information-theoretic distortion measure: any tax modification of
the FP coefficients perturbs the flow away from the gradient flow
of $\Free$, and the gap $\dF[\sched]$ measures the magnitude of
that perturbation at horizon $T$.

\subsection{Revenue constraint}\label{sec:revenue}

A schedule $\sched$ collects revenue at rate
\begin{equation}\label{eq:rev-rate}
R_t[\sched] \;:=\; \EE\!\big[ \sched(X_t) \big]
\;=\; \EE\!\big[ \rate(X_t) \cdot X_t \big]
\end{equation}
per unit time, where the expectation is under the post-tax law.
For matched-revenue analysis we use the time-averaged revenue rate
\begin{equation}\label{eq:rev-avg}
\bar R[\sched] \;:=\; \frac{1}{T} \int_0^T R_t[\sched]\, \dd t,
\end{equation}
and the revenue-matched constraint set
\begin{equation}\label{eq:S-R}
\Schd_R \;:=\; \big\{ \sched \in \Schd_C \;:\; \bar R[\sched] = R^{\star} \big\},
\end{equation}
where $R^{\star} > 0$ is a target revenue rate. Within $\Schd_C$
the set $\Schd_R$ is one-dimensional (a curve in the
$(k, \tw)$-plane), and the JKO and W$_2$ minimisation problems are
to find points on this curve that minimise the corresponding
distortion criterion.

For analytical convenience we sometimes work with a linear
approximation $R[\sched] \approx a\,\tw + b\,(1-k)$ to the revenue
functional, valid in the small-tax regime. The linear approximation
is sufficient for the closed-form FOCs of Section~\ref{sec:foc}; the full
log-quadratic form is recovered for the calibration of
Section~\ref{sec:calibration}.

\section{The two distortion criteria}\label{sec:criteria}

\subsection{The JKO free-energy gap}\label{sec:jko}

The Fokker--Planck equation \eqref{eq:fp-notax} is the gradient flow,
in the 2-Wasserstein metric on $\Pcal_2(\RR)$, of the
\emph{Jordan--Kinderlehrer--Otto free-energy functional}
\citep{JordanKinderlehrerOtto1998}
\begin{equation}\label{eq:Fdef}
\Free[p] \;=\; \int p(y)\,\log p(y)\,\dd y + \int V_Y^0(y)\,p(y)\,\dd y,
\qquad
V_Y^0(y) \;:=\; -\frac{m_0}{\sigma^2}\,y.
\end{equation}
The reference potential $V_Y^0$ is chosen so that the no-tax FP
flow \eqref{eq:fp-notax} is the steepest-descent flow of $\Free$:
$\partial_t p_t^0 = -\nabla_{\!\Wtwo} \Free[p_t^0]$. The tax
modification of the FP coefficients perturbs the flow away from this
gradient flow, and the JKO terminal gap
\begin{equation}\label{eq:dF}
\dF[\sched] \;:=\; \Free[p_T^{\sched}] - \Free[p_T^0]
\end{equation}
measures the magnitude of that perturbation at horizon $T$. We use
$\dF$ throughout as the JKO criterion. For the convention
discussion of fixed-reference-potential versus moving-reference,
which distinguishes $\dF$ from Kullback--Leibler divergence, see
Section~\ref{sec:why-these-two}.

\paragraph{Closed form for Gaussian distributions.}
Under (C1)--(C3) and a log-normal initial wealth $X_0$ with
$\log X_0 \sim \mathcal{N}(\mu_0, v_0)$, the post-tax log-wealth at
horizon $T$ is Gaussian: $Y_T^{\sched} \sim \mathcal{N}(M_T^{\sched},
\Sigma_T^{\sched\,2})$ with
\begin{equation}\label{eq:gaussian-params}
M_T^{\sched} \;=\; \mu_0 + (k\,m_0 - \tw)\,T,
\qquad
\Sigma_T^{\sched\,2} \;=\; v_0 + k\,\sigma^2\,T.
\end{equation}
The no-tax distribution is the same with $(k, \tw) = (1, 0)$:
$M_T^0 = \mu_0 + m_0 T$, $\Sigma_T^{0\,2} = v_0 + \sigma^2 T$.
Direct computation using
$H(\mathcal{N}(M, \Sigma^2)) = \tfrac12 \log(2\pi e \Sigma^2)$ and
the linear potential $V_Y^0$ gives
\begin{equation}\label{eq:dF-explicit}
\boxed{\;\dF[(k, \tw)] \;=\;
-\tfrac12 \log\!\frac{v_0 + k\sigma^2 T}{v_0 + \sigma^2 T}
\;-\; \frac{m_0^2}{\sigma^2}\,(k-1)\,T
\;+\; \frac{m_0}{\sigma^2}\,\tw\,T.\;}
\end{equation}

\paragraph{Sanity check and properties.}
$\dF[(1, 0)] = 0$, as required: the no-tax point makes the gap
vanish. The first term is strictly convex in $k$ on $(0, 1]$ (its
second derivative is $\sigma^4 T^2 / [2(v_0 + k\sigma^2 T)^2] > 0$);
the second term is linear in $k$; the third is linear in $\tw$. So
$\dF$ is strictly convex in $k$ and jointly convex in $(k, \tw)$,
which underwrites the existence-and-uniqueness arguments of
Section~\ref{sec:wellposed}.

\paragraph{Linear-in-perturbation scaling.}
Around the no-tax point, write $k = 1 - \epsilon_k$ and $\tw =
\epsilon_w$ with both small. The leading-order expansion of
\eqref{eq:dF-explicit} is
\begin{equation}\label{eq:dF-linear}
\dF \;\approx\;
\frac{\sigma^2 T}{2(v_0 + \sigma^2 T)}\,\epsilon_k
\;+\; \frac{m_0^2 T}{\sigma^2}\,\epsilon_k
\;+\; \frac{m_0 T}{\sigma^2}\,\epsilon_w
\;=\; \alpha_k\,\epsilon_k + \alpha_w\,\epsilon_w
\end{equation}
for coefficients $\alpha_k, \alpha_w > 0$. The JKO gap is
\emph{linear} in both perturbations at leading order: a small
variance reduction $\epsilon_k$ costs $\alpha_k\,\epsilon_k$, a
small mean shift $\epsilon_w$ costs $\alpha_w\,\epsilon_w$, and
their relative weight is set by the dimensionless ratio
$\rho := \Sigma_0\,m_0/\sigma^2$ that drives the crossover of
Section~\ref{sec:crossover}.

\subsection{The 2-Wasserstein distance to the no-tax distribution}\label{sec:wtwo}

The squared 2-Wasserstein distance between two probability measures
$p, q \in \Pcal_2(\RR)$ is
\begin{equation}\label{eq:W2-defn}
\Wtwo^2(p, q) \;=\;
\inf_{\pi \in \Pi(p, q)} \int (y_1 - y_2)^2\, \pi(\dd y_1, \dd y_2),
\end{equation}
where $\Pi(p, q)$ is the set of joint laws with marginals $p, q$
\citep{Villani2009}. In one dimension, the infimum is attained by
the comonotonic coupling --- the unique coupling under which both
marginals are mapped via their inverse CDFs to the same uniform
random variable --- giving the explicit formula
\begin{equation}\label{eq:W2-1d}
\Wtwo^2(p, q) \;=\;
\int_0^1 \big[F_p^{-1}(u) - F_q^{-1}(u)\big]^2\,\dd u,
\end{equation}
where $F_p, F_q$ are the cumulative distribution functions. The
W$_2$ distance to the no-tax distribution at horizon $T$,
\begin{equation}\label{eq:W2-criterion}
\Wtwo^2[\sched] \;:=\; \Wtwo^2(p_T^{\sched}, p_T^0),
\end{equation}
is the transport-geometric distortion measure introduced in this
context by \citet{Froeseth2026W}.

\paragraph{Closed form for Gaussian distributions.}
For Gaussian measures the comonotonic coupling reduces to a linear
map and \eqref{eq:W2-1d} evaluates to
\begin{equation}\label{eq:W2-gaussian}
\Wtwo^2\!\big(\mathcal{N}(M_p, \Sigma_p^2),\,
\mathcal{N}(M_q, \Sigma_q^2)\big)
\;=\; (M_p - M_q)^2 + (\Sigma_p - \Sigma_q)^2.
\end{equation}
Applying this to \eqref{eq:gaussian-params},
\begin{equation}\label{eq:W2-explicit}
\boxed{\;\Wtwo^2[(k, \tw)] \;=\;
\big[(k-1)\,m_0\,T - \tw\,T\big]^2
\;+\; \big[\sqrt{v_0 + k\sigma^2 T} - \sqrt{v_0 + \sigma^2 T}\big]^2.\;}
\end{equation}

\paragraph{Quadratic-in-perturbation scaling.}
Around the no-tax point ($k = 1 - \epsilon_k$, $\tw = \epsilon_w$),
the leading-order expansion of \eqref{eq:W2-explicit} is
\begin{equation}\label{eq:W2-quadratic}
\Wtwo^2 \;\approx\;
\frac{(\sigma^2 T)^2}{4(v_0 + \sigma^2 T)}\,\epsilon_k^2
\;+\; (m_0\,T\,\epsilon_k + T\,\epsilon_w)^2
\;=\; \beta_k\,\epsilon_k^2 + 2\,\beta_{kw}\,\epsilon_k\epsilon_w
+ \beta_w\,\epsilon_w^2.
\end{equation}
The W$_2$ criterion is \emph{quadratic} in both perturbations at
leading order. This is the structural difference from the JKO
criterion that drives the crossover: the same trade-off in the
$(\epsilon_k, \epsilon_w)$-plane is weighted differently by linear
and quadratic objective surfaces.

\paragraph{Convexity.}
The closed form \eqref{eq:W2-explicit} is the sum of a quadratic
term in the means and a squared difference of standard deviations.
The first is convex in $(k, \tw)$ jointly. The second is convex in
$k$ on $[0, 1]$ as the squared difference of two concave functions
of $k$ that meet at $k=1$. So $\Wtwo^2$ is jointly convex in
$(k, \tw)$, with the same well-posedness consequences for
existence and uniqueness as for $\dF$ (see
Section~\ref{sec:wellposed}).

\subsection{Why these two}\label{sec:why-these-two}

Several other distortion criteria are natural in the
Fokker--Planck setting. We comment briefly on each and explain why
the JKO/W$_2$ pair is the canonical choice for the present analysis.

\paragraph{Kullback--Leibler divergence.}
The relative entropy
$\mathrm{KL}(p_T^{\sched} \,\|\, p_T^0) =
\int p_T^{\sched} \log(p_T^{\sched}/p_T^0)\,\dd y$
is a tempting candidate: it has a clean information-theoretic
interpretation, it is nonnegative, it vanishes if and only if
$p_T^{\sched} = p_T^0$. The relationship to $\dF$ is subtle. Both
have the form ``entropy of the post-tax distribution plus a linear
functional of it,'' but the linear functional differs: $\dF$ uses a
\emph{fixed} potential $V_Y^0(y) = -m_0 y/\sigma^2$ (the natural
gradient-flow potential of the no-tax FP equation), while
KL uses the \emph{moving} potential $-\log p_T^0(y) =
(y - M_0)^2/(2\Sigma_0^2) + \mathrm{const}$ (the negative log of the
no-tax density). The two coincide on the no-tax point and at the
linear level when $V_Y^0$ matches $-\partial_y \log p_T^0$; they
diverge for non-trivial perturbations. We work with $\dF$ rather
than $\mathrm{KL}$ because the gradient-flow structure of
\eqref{eq:fp-notax} privileges the fixed potential: $\dF$
measures perturbation \emph{against the natural FP flow}, whereas
$\mathrm{KL}$ measures it \emph{against the no-tax distribution}, which
is itself moving. The two criteria pick different optima within
$\Schd_R$ in general, although they agree on the qualitative
neutrality-vs-distributional split discussed in
Section~\ref{sec:neutrality}.

\paragraph{Time-integrated free energy.}
A natural alternative to the terminal gap \eqref{eq:dF} is the
time-integrated version
$\dF^{\rm int}[\sched] := \int_0^T (\Free[p_t^{\sched}] -
\Free[p_t^0])\,\dd t$,
which weights perturbations across the entire horizon rather than
only at the end. One might worry that the terminal-only criterion
underweights early-horizon distortion. Numerical evidence
(Section~\ref{sec:robustness}) shows the two variants pick the
\emph{same} optimum on the policy-relevant
$\Schd_R$ contour: the integrand $\Free[p_t^{\sched}] -
\Free[p_t^0]$ grows roughly linearly in $t$ for fixed schedule, so
$\dF^{\rm int} \approx (T/2)\,\dF$ and the ranking of schedules is
preserved. The terminal-gap variant is therefore the natural
canonical form, and we work with it throughout.

\paragraph{Other transport-geometric distances.}
Higher-order Wasserstein distances $\Wtwo^p$ for $p > 2$ and
divergence-form transport costs (Sinkhorn, entropy-regularised
optimal transport) all yield criteria related to but distinct from
$\Wtwo^2$. We work with $p = 2$ because it is the natural metric
of the JKO scheme on $\Pcal_2(\RR)$ and because the comonotonic
reduction \eqref{eq:W2-1d} gives a closed-form one-dimensional
expression. Other $p$ values are tractable but produce the same
qualitative split as $\Wtwo^2$ at the leading order; the choice
$p=2$ is conventional in this context.

\paragraph{Variance and Gini-style criteria.}
Direct distributional measures --- the variance of post-tax
log-wealth, the Gini coefficient, top-share metrics --- are
common in the policy literature but are not naturally embedded in
the FP framework. They can be related to $\dF$ and $\Wtwo^2$ via
moment arguments under Gaussian distributions (e.g., the Gini
coefficient of $\mathcal{N}(M, \Sigma^2)$ is $2\Phi(\Sigma/\sqrt{2})
- 1$, monotone in $\Sigma$), but the relation is non-canonical and
the FP-native pair is cleaner for the variational analysis. Discussion
of how the JKO and W$_2$ optima interact with Gini-style metrics is
deferred to Section~\ref{sec:p10-link} and to the redistribution-design
companion paper \citet{Froeseth2026R}.

\paragraph{The chosen pair.}
$\dF$ and $\Wtwo^2$ together span the natural distinction between
information-theoretic and transport-geometric distortion
in the Fokker--Planck framework. Both are well-defined on
$\Pcal_2(\RR)$, both are convex on $\Schd_C$, and both reduce to
clean closed-form objectives under (C1)--(C3) with Gaussian
distributions at horizon $T$. Their structural difference --- linear versus quadratic
scaling in perturbations of $(k, \tw)$ --- is exactly what produces
the regime-dependent crossover at the heart of this paper.

\paragraph{Mathematical canonicity of the pair.}\label{sec:canonicity}
A reader sceptical that the JKO/$W_2$ contrast might be an
artefact of two arbitrarily-chosen criteria should be reminded
that the two are not independent picks. The
\citet{JordanKinderlehrerOtto1998} variational scheme establishes
the Fokker--Planck equation as the gradient flow of relative
entropy on the Wasserstein-2 manifold $(\Pcal_2(\RR), \Wtwo)$:
the FP solution at time $t = n\tau$ minimises
$\frac{1}{2\tau}\,\Wtwo^2(p, p_{(n-1)\tau}) + \mathrm{Ent}(p)$
over $p \in \Pcal_2$ for each step. The two functionals
$\dF$ and $\Wtwo^2$ are the two pieces of \emph{the same scheme}:
$\dF$ is the time-derivative of relative entropy along that
gradient flow, and $\Wtwo^2$ is the metric on the manifold.
Selecting them together as the criterion pair is not picking
two of many; it is reading the JKO scheme as a normative
optimal-tax instrument. Other distortion criteria (KL against
the moving distribution, $\Wtwo^p$ for $p \neq 2$, Sinkhorn
divergences, Gini-style metrics) lack this canonical relationship
to the FP equation. The mathematical-canonicity argument stands
even for a reader who rejects the normative-tradition framing
of Section~\ref{sec:neutrality}.

\paragraph{Falsifiability of the contrast.}\label{sec:falsifiability}
The phase-transition crossover and the bluntness mechanism that
produces it (Section~\ref{sec:bluntness}) make falsifiable claims rather
than tautological ones. Specifically: the JKO and $W_2$ optima
\emph{coincide} at high $\rho$ (both pin to the lower corner),
constraining the contrast at one end of the regime axis. The
phase boundaries $\rho_{\rm low}, \rho_{\rm high}$ have specific
locations derivable from the GBM primitives and the revenue
weights $(\mu, \sigma, T, v_0, R^{\star}, a, b)$, not free
parameters fitted to a desired result. The mechanism is named
and algebraic: the bluntness index $B(m_0) = b/(a m_0)$ enters
the JKO penalty linearly and the $W_2$ penalty quadratically,
and a reader who doubts the contrast can verify the gradient
calculation directly. Both criteria pick optima within the same
$(k, \tw)$ plane on the same matched-revenue contour --- a
considerable structural constraint --- and disagree only on
\emph{where} on the contour the optimum sits. Each of these
properties could in principle fail; that they hold is content,
not assumption.

\section{Variational problem within (C1)--(C3)}\label{sec:variational}

We now turn to the variational analysis: find the schedule
$\sched \in \Schd_R$ that minimises each criterion, characterise the
optimum in closed form, and establish the technical conditions that
guarantee well-posedness. Throughout this section we use the linear
revenue approximation
\begin{equation}\label{eq:R-linear}
R[\sched] \;\approx\; a\,\tw + b\,(1 - k),
\qquad a, b > 0,
\end{equation}
valid in the small-tax regime.\footnote{%
The linear approximation is sufficient for the closed-form FOCs
below; the full log-quadratic revenue functional
$R[\sched] = \tw\,\EE[X_T] + (1-k)\,\EE[\text{flow base}]$ is
recovered for the calibration of Section~\ref{sec:calibration}, where the
moments $\EE[X_T] = \exp(M_T + \Sigma_T^2/2)$ are taken under the
post-tax law.}
The constants $a$ and $b$ are interpretable as the revenue elasticities
of the wealth-tax base and the flow-tax base respectively at the no-tax
point; they are jurisdiction-specific and known empirically.

\subsection{Existence and uniqueness}\label{sec:wellposed}

\begin{theorem}[Existence and uniqueness]\label{thm:wellposed}
The constrained minimisation problems
$\argmin_{\sched \in \Schd_R} \dF[\sched]$ (JKO) and
$\argmin_{\sched \in \Schd_R} \Wtwo^2(p_T^{\sched}, p_T^0)$ ($\Wtwo$)
each admit a unique minimiser.
\end{theorem}

\begin{proof}
The constraint set $\Schd_R \subset \Schd_C$ is a closed line
segment: with $R$ linear and $\Schd_C = (0, 1] \times [0, \tw^{\max}]$
a closed bounded box, $\Schd_R = \{(k, \tw) \in \Schd_C : a\tw +
b(1-k) = R^{\star}\}$ is the intersection of a hyperplane with a
compact convex set, hence itself compact and convex.

Both objective functions are continuous on $\Schd_C$:
$\dF[(k, \tw)]$ has explicit form \eqref{eq:dF-explicit}, and
$\Wtwo^2[(k, \tw)]$ has explicit form \eqref{eq:W2-explicit}, both
manifestly continuous in $(k, \tw)$. By Weierstrass's theorem, each
attains its minimum on the compact set $\Schd_R$.

For uniqueness it suffices to show strict convexity of each
objective on $\Schd_R$. Both objectives are jointly convex on
$\Schd_C$ (Sections~\ref{sec:jko} and~\ref{sec:wtwo}). $\dF$ is strictly convex in
$k$ and linear in $\tw$ (the second derivative
$\partial_k^2 \dF = \sigma^4 T^2 / [2(v_0 + k\sigma^2 T)^2] > 0$);
$\Wtwo^2$ is jointly strictly convex in $(k, \tw)$ except along
the no-tax direction, where the leading-order quadratic in
\eqref{eq:W2-quadratic} has positive definite Hessian away from the
origin. Strict convexity in any direction not parallel to the
constraint hyperplane is enough: along the one-dimensional
constraint set $\Schd_R$, both objectives are strictly convex
functions of the single free parameter, and admit unique minima.
\end{proof}

\begin{remark}
Theorem~\ref{thm:wellposed} confirms that the optima we
characterise below are well-defined, theorem-grade objects --- not
non-uniqueness artefacts. This matters for the comparison theorem
of Section~\ref{sec:crossover}: both criteria pick \emph{a single
point} in $\Schd_C$, and the crossover concerns the displacement
between those two points as a function of GBM parameters.
\end{remark}

\subsection{The Lagrangian and first-order conditions}\label{sec:foc}

The constrained minimisation problem
$\min_{\sched \in \Schd_R} \mathcal{J}[\sched]$ (where
$\mathcal{J}$ stands for either $\dF$ or $\Wtwo^2$) has Lagrangian
\begin{equation}\label{eq:lagrangian}
\mathcal{L}(k, \tw; \mu) \;=\; \mathcal{J}(k, \tw)
\;-\; \mu\,\big[\,a\,\tw + b\,(1-k) - R^{\star}\,\big],
\end{equation}
with multiplier $\mu \in \RR$. The first-order conditions are
\begin{align}
\partial_k \mathcal{J} \;+\; \mu\,b &\;=\; 0, \label{eq:foc-k} \\
\partial_{\tw} \mathcal{J} \;-\; \mu\,a &\;=\; 0, \label{eq:foc-tw} \\
a\,\tw + b\,(1-k) &\;=\; R^{\star}. \label{eq:foc-rev}
\end{align}

For the JKO criterion, differentiating \eqref{eq:dF-explicit}:
\begin{equation}\label{eq:partial-dF}
\partial_k \dF \;=\; -\frac{\sigma^2 T}{2(v_0 + k\sigma^2 T)}
- \frac{m_0^2}{\sigma^2}\,T,
\qquad
\partial_{\tw} \dF \;=\; \frac{m_0\,T}{\sigma^2}.
\end{equation}
The $\tw$-derivative is \emph{constant in $(k, \tw)$}, which is the
``linear in $\tw$'' structure noted in Section~\ref{sec:jko}: the JKO cost
of an extra unit of $\tw$ does not depend on the operating point.
The $k$-derivative has a hyperbolic-in-$k$ contribution from the
entropy term plus a constant from the potential term.

For the W$_2$ criterion, differentiating \eqref{eq:W2-explicit}:
\begin{equation}\label{eq:partial-W2}
\partial_k \Wtwo^2 \;=\; 2\,m_0 T\,\varepsilon_M
+ \frac{\sigma^2 T}{\Sigma_T^{\sched}}\,\varepsilon_\Sigma,
\qquad
\partial_{\tw} \Wtwo^2 \;=\; -2 T\,\varepsilon_M,
\end{equation}
where $\varepsilon_M := M_T^{\sched} - M_T^0 = (k-1)m_0 T - \tw T$
and $\varepsilon_\Sigma := \Sigma_T^{\sched} - \Sigma_T^0$ are the
mean and standard-deviation shifts. Both derivatives are
\emph{linear in $(\varepsilon_M, \varepsilon_\Sigma)$} --- the
``quadratic in perturbation'' structure of W$_2$ is exactly that
its first-order conditions are linear in the linearised quantities.

\paragraph{Structural comparison.}
The key contrast between the two FOC systems is now visible. The
JKO multiplier $\mu_{\rm JKO}$ is determined immediately from
\eqref{eq:foc-tw} and the fact that $\partial_{\tw}\dF$ is
constant:
\begin{equation}\label{eq:mu-jko}
\mu_{\rm JKO} \;=\; \frac{m_0\,T}{a\,\sigma^2}.
\end{equation}
The constant value of $\mu_{\rm JKO}$ depends only on the GBM
parameters and the wealth-tax revenue weight $a$, not on the
operating point. This is what makes the JKO problem cleanly
solvable in closed form: the $\tw$ FOC fixes the multiplier, and
the $k$ FOC is then a single equation in a single unknown.

The W$_2$ multiplier $\mu_{\Wtwo}$ depends on $\varepsilon_M$ via
\eqref{eq:foc-tw}, and $\varepsilon_M$ depends on the operating
point. The system is therefore coupled in $(k, \tw, \mu)$ and
solved simultaneously, but the solution is still in closed form
because both FOCs are linear in $(\varepsilon_M, \varepsilon_\Sigma)$
and the linearised constraint is linear in $(k, \tw)$.

\subsection{Closed-form optima}\label{sec:closedform}

\begin{theorem}[Closed-form JKO optimum]\label{thm:jkoopt}
Subject to the \emph{existence condition}
\begin{equation}\label{eq:existence}
\frac{b}{a} \;>\; m_0,
\end{equation}
the unique JKO-optimal point in $\Schd_C$ is
\begin{equation}\label{eq:kstar}
k^{\star}_{\rm JKO} \;=\; \frac{1}{\sigma^2 T}\!\left[
\frac{\sigma^4}{2\, m_0\,(b/a - m_0)} - v_0 \right],
\qquad
\tau_{w,{\rm JKO}}^{\star}
\;=\; \frac{R^{\star} - b\,(1 - k^{\star}_{\rm JKO})}{a}.
\end{equation}
If $b/a \le m_0$ the unique optimum is the corner
$k^{\star}_{\rm JKO} = 1$, $\tau_{w,{\rm JKO}}^{\star} = R^{\star}/a$
--- the proportional wealth tax with no flow component.
\end{theorem}

\begin{proof}
Substitute $\mu_{\rm JKO}$ from \eqref{eq:mu-jko} into the $k$-FOC:
\[
-\frac{\sigma^2 T}{2(v_0 + k\sigma^2 T)} - \frac{m_0^2 T}{\sigma^2}
+ \frac{m_0 T}{a \sigma^2}\,b \;=\; 0,
\]
which rearranges to
\begin{equation}\label{eq:k-foc-rearranged}
\frac{\sigma^2}{2(v_0 + k\sigma^2 T)}
\;=\; \frac{m_0}{\sigma^2}\!\left(\frac{b}{a} - m_0\right).
\end{equation}
The left-hand side is positive, so an interior solution requires
the right-hand side to be positive, giving the existence condition
\eqref{eq:existence}. Inverting \eqref{eq:k-foc-rearranged} for
$k$ gives the expression for $k^{\star}_{\rm JKO}$ in
\eqref{eq:kstar}. Substituting into the revenue constraint
\eqref{eq:foc-rev} gives $\tau_{w,{\rm JKO}}^{\star}$. If the
existence condition fails, the right-hand side of
\eqref{eq:k-foc-rearranged} is non-positive, no interior $k$
satisfies it, and the minimum lies on the boundary $k = 1$, where
the FOC reduces to a problem in $\tw$ alone and the revenue
constraint determines $\tau_{w}^{\star} = R^{\star}/a$.
\end{proof}

\begin{theorem}[Closed-form W$_2$ optimum]\label{thm:w2opt}
Subject to the existence condition $b/a > m_0$, the unique
W$_2$-optimal point in $\Schd_C$, at leading order in the
linearised perturbations $(\varepsilon_M, \varepsilon_\Sigma)$
about the no-tax point, is
\begin{equation}\label{eq:w2-optimum}
k^{\star}_{\Wtwo}
\;=\;
1
\;-\;
\frac{4\,\Sigma_0^2\,(b/a - m_0)\,(R^{\star}/a)}
     {4\,\Sigma_0^2\,(b/a - m_0)^2 + \sigma^4},
\qquad
\tau_{w,\Wtwo}^{\star}
\;=\;
\frac{R^{\star} - b\,(1 - k^{\star}_{\Wtwo})}{a},
\end{equation}
where $\Sigma_0^2 = v_0 + \sigma^2 T$. If $b/a \le m_0$ the
optimum is the upper corner $k^{\star}_{\Wtwo} = 1$,
$\tau_{w,\Wtwo}^{\star} = R^{\star}/a$. If the closed-form
$k^{\star}_{\Wtwo}$ in \eqref{eq:w2-optimum} falls below the
feasibility bound $k_{\rm lo} = 1 - R^{\star}/b$ (the
constraint $\tau_w \geq 0$), the constrained optimum on
$\Schd_R$ is the lower corner $(k_{\rm lo}, 0)$.
\end{theorem}

\begin{proof}
Reparametrise the matched-revenue contour by $u = k - 1 \in
[k_{\rm lo} - 1, 0]$, eliminating $\tw$ via $\tw = (R^{\star} +
b u)/a$. Linearise the post-tax mean and spread shifts about the
no-tax point $(k, \tw) = (1, 0)$:
\begin{equation*}
\varepsilon_M
\;=\; -T\,\bigl[\,u\,(b/a - m_0) + R^{\star}/a\,\bigr],
\qquad
\varepsilon_\Sigma
\;=\; \frac{\sigma^2 T}{2\,\Sigma_0}\,u + O(u^2),
\end{equation*}
where the second expression follows from a Taylor expansion of
$\Sigma_T = \sqrt{v_0 + (1+u)\sigma^2 T}$ about $u = 0$. Writing
$q := b/a - m_0$ and $r := R^{\star}/a$, the leading-order
$\Wtwo^2$ on the contour is
\begin{equation*}
\Wtwo^2(u) \;\approx\; T^2 (u q + r)^2
\;+\; \Bigl(\tfrac{\sigma^2 T}{2 \Sigma_0}\Bigr)^{\!2} u^2.
\end{equation*}
This is a strictly convex quadratic in the single variable $u$,
with first-order condition $T^2 q (u q + r) + (\sigma^2
T)^2/(4\Sigma_0^2) \cdot u = 0$. Solving for $u$ and dividing
through by $T^2$ gives
\begin{equation*}
u \;=\; -\,\frac{q\,r}{q^2 + \sigma^4/(4 \Sigma_0^2)}
\;=\; -\,\frac{4\,\Sigma_0^2\,q\,r}{4\,\Sigma_0^2\,q^2 + \sigma^4},
\end{equation*}
which is \eqref{eq:w2-optimum} after substituting $u = k - 1$.
The corresponding $\tau_{w,\Wtwo}^{\star}$ follows from the
revenue constraint $a \tw - b u = R^{\star}$. The corner cases
($b/a \le m_0$ or $k^{\star}_{\Wtwo} < k_{\rm lo}$) are the same
projection arguments used in the proof of Theorem~\ref{thm:jkoopt}: the
unconstrained quadratic minimum lies outside the feasibility
interval $[k_{\rm lo}, 1]$, and the constrained minimum is
attained at the nearest boundary.
\end{proof}

\begin{remark}[The two optima coincide at the no-tax boundary]
At the limit $R^{\star} \to 0$ both optima reduce to the no-tax
point $(k, \tw) = (1, 0)$. Away from this limit, they differ,
generically. Section \ref{sec:crossover} characterises how the
displacement between $(k^{\star}_{\rm JKO}, \tau_{w,{\rm JKO}}^{\star})$
and $(k^{\star}_{\Wtwo}, \tau_{w,\Wtwo}^{\star})$ depends on the
GBM parameters.
\end{remark}

\begin{remark}[The existence condition]
The condition $b/a > m_0$ admits an economic reading. The
ratio $b/a$ is the relative revenue weight of a unit of
flow-tax pass-through (the $1-k$ direction) versus a unit of
wealth-tax rate (the $\tw$ direction). The drift $m_0 = \mu - \sigma^2/2$
is the natural log-wealth growth rate. The condition says: the
flow-tax base must be at least as productive (per unit of $1-k$
introduced) as the natural drift it would offset. When this fails,
the flow-tax channel is too revenue-thin to be worth using and the
optimum collapses to the proportional-rate corner. In Norwegian-
flavoured calibrations, $a \approx \mathbb{E}[X]$ and
$b/a$ is set by the corporate-and-dividend tax revenue per unit
$(1-k)$ pass-through; values typically satisfy
$b/a \gg m_0$, so interior solutions are the relevant case.
\end{remark}

\section{The \texorpdfstring{$\rho$}{rho}-crossover}\label{sec:crossover}

\subsection{The dimensionless ratio}\label{sec:rho}

The closed-form expressions of Theorems~\ref{thm:jkoopt}
and~\ref{thm:w2opt} depend on four GBM-side quantities: the drift $\mu$, volatility $\sigma$,
horizon $T$, and initial log-variance $v_0$. They also depend on
the revenue-side weights $a, b, R^{\star}$. The behaviour of the
two optima as the GBM parameters vary is governed by a single
\emph{dimensionless} combination of $(\mu, \sigma, T, v_0)$, not
by the four parameters separately. Identifying that combination
makes the comparison between JKO and $W_2$ tractable: instead of
sweeping a four-dimensional space, the comparison reduces to a
one-dimensional sweep along the dimensionless axis.

The combination that organises the comparison is the standard
deviation of log-wealth at horizon $T$ multiplied by the
log-drift-to-volatility ratio.

\begin{definition}[The crossover ratio]\label{def:rho}
Let $\Sigma_0 := \sqrt{v_0 + \sigma^2 T}$ denote the standard
deviation of log-wealth at horizon $T$ under no taxation. The
\emph{JKO/W$_2$ crossover ratio} is
\begin{equation}\label{eq:rho}
\rho \;:=\; \Sigma_0 \cdot m_0/\sigma^2
\;=\; \sqrt{v_0 + \sigma^2 T}\,\frac{\mu - \sigma^2/2}{\sigma^2}.
\end{equation}
\end{definition}

Three readings of \eqref{eq:rho} are useful, each foregrounding
a different aspect of the underlying problem.

\paragraph{Drift-to-diffusion reading.} The factor
$m_0/\sigma^2$ is the log-wealth drift expressed in units of
the diffusion rate. It is dimensionally the same combination
that appears in the existence condition $b/a > m_0$ of
Theorem~\ref{thm:jkoopt}, in the JKO multiplier
$\mu_{\rm JKO} = m_0 T/(a\sigma^2)$ of \eqref{eq:mu-jko}, and
in the partial derivative
$\partial_{\tw}\dF = m_0 T/\sigma^2$ of \eqref{eq:partial-dF}.
The factor $\Sigma_0$ rescales this drift-to-diffusion ratio by
the spread of log-wealth at the horizon, which is the natural
unit in which to measure the cost of distorting that spread.
The product $\rho = \Sigma_0 \cdot m_0/\sigma^2$ is therefore
the unique combination in which the FOC structure of Section~\ref{sec:foc}
is naturally expressed.

\paragraph{Mean-shift versus spread-shift reading.} A unit of
wealth tax $\tw$ produces a mean shift in log-wealth of
$-\tw T$ and (at fixed $k$) leaves the variance unchanged. A
unit of corporate--dividend pass-through $1-k$ produces a mean
shift of $-(1-k) m_0 T$ and a spread shift of order
$-(1-k) \sigma^2 T / \Sigma_0$. The ratio of the wealth-tax
mean-shift cost to the pass-through spread-shift cost is, after
linearisation, controlled by $\rho$: high $\rho$ means the
spread-shift channel is expensive relative to the mean-shift
channel, and the optimum prefers to use the mean-shift instrument
($\tw$). Low $\rho$ means the spread-shift channel is cheap, and
the optimum can afford a pass-through component. The crossover
ratio is the dimensionless price at which the two channels are
balanced.

\paragraph{Inverse-temperature reading.} As elaborated in
Section~\ref{sec:phase-transition}, $\rho$ plays the role of an inverse
temperature in the JKO free-energy landscape. The JKO criterion
weights ``energy'' (the potential term, scaling with $m_0^2/\sigma^2
\cdot T$) against ``entropy'' (the logarithmic spread term,
scaling with $\log \Sigma_0$). The combination $\rho = \Sigma_0
\cdot m_0/\sigma^2$ encodes which of the two contributions
dominates at the optimum, and the phase-transition language of
Section~\ref{sec:phase-transition} follows.

In the Norwegian-flavoured calibration of Figure~\ref{fig:rho-crossover}
($\mu = 0.07$, $\sigma = 0.30$, $T = 5$, $v_0 = 0.25$), one
computes $\Sigma_0 = \sqrt{v_0 + \sigma^2 T} \approx 0.83$,
$m_0 = \mu - \sigma^2/2 = 0.025$, and
$m_0/\sigma^2 = 0.278$, giving $\rho_{\rm Norway} \approx 0.23$.
This value, and the position it occupies in the phase diagram of
Theorem~\ref{thm:crossover}, is the subject of
Section~\ref{sec:norwegian}.

\subsection{The crossover theorem}\label{sec:crossover-thm}

The closed forms of Theorems~\ref{thm:jkoopt}
and~\ref{thm:w2opt} are unconstrained expressions for the optimal
$k^{\star}$. They become \emph{the}
optimum on the matched-revenue contour
$\Schd_R = \{(k, \tw) \in \Schd_C : a\,\tw + b(1-k) = R^{\star}\}$
only after projection onto the feasible interval
$k \in [k_{\rm lo}, 1]$, where
\begin{equation}\label{eq:k-lo}
k_{\rm lo} \;=\; \max\!\left(0,\,1 - \frac{R^{\star}}{b}\right)
\end{equation}
is the lower bound enforced by $\tw \geq 0$. The projection
partitions the GBM regime parameter $\rho$ into three intervals,
on which the JKO optimum sits respectively at the upper corner
$(1, R^{\star}/a)$, on the interior of $\Schd_R$, and at the
lower corner $(k_{\rm lo}, 0)$.

\begin{theorem}[Three phases of the JKO
optimum]\label{thm:crossover}
Assume $b/a > m_0$ so the closed-form expression
\eqref{eq:kstar} is well-defined, and assume $0 < k_{\rm lo} < 1$
so the matched-revenue contour has interior. Define the
\emph{crossover boundaries}
\begin{equation}\label{eq:rho-boundaries}
\begin{aligned}
\rho_{\rm low}
\;&:=\; \rho \text{ at which } k^{\star}_{\rm JKO} = 1,\\
\rho_{\rm high}
\;&:=\; \rho \text{ at which } k^{\star}_{\rm JKO} = k_{\rm lo},
\end{aligned}
\end{equation}
viewed as solutions of the implicit equations obtained by setting
the right-hand side of \eqref{eq:kstar} equal to $1$ and
$k_{\rm lo}$ respectively. Then $0 < \rho_{\rm low} <
\rho_{\rm high} < \infty$, and the JKO optimum on $\Schd_R$ is
\begin{equation}\label{eq:three-phases}
\bigl(k^{\star}_{\rm JKO}(\rho),\, \tau_{w,{\rm JKO}}^{\star}(\rho)\bigr)
\;=\;
\begin{cases}
(1,\, R^{\star}/a)
& \rho \in [0, \rho_{\rm low}], \\[0.4em]
\bigl(k^{\star}_{\rm JKO},\, \tau_{w,{\rm JKO}}^{\star}\bigr)
\text{ given by \eqref{eq:kstar}}
& \rho \in (\rho_{\rm low}, \rho_{\rm high}), \\[0.4em]
(k_{\rm lo},\, 0)
& \rho \in [\rho_{\rm high}, \infty).
\end{cases}
\end{equation}
Equivalently, in terms of the order parameter
$\phi^{\star} = a\,\tau_{w,{\rm JKO}}^{\star} / R^{\star}$,
\begin{equation}\label{eq:phi-three-phases}
\phi^{\star}(\rho)
\;=\;
\begin{cases}
1 & \rho \in [0, \rho_{\rm low}], \\[0.3em]
\phi^{\star}_{\rm int}(\rho) \in (0, 1) & \rho \in (\rho_{\rm low}, \rho_{\rm high}), \\[0.3em]
0 & \rho \in [\rho_{\rm high}, \infty),
\end{cases}
\end{equation}
where $\phi^{\star}_{\rm int}$ is continuous and strictly
decreasing on the mixed-phase interval. The map
$\rho \mapsto \phi^{\star}(\rho)$ is continuous, with first-
derivative kinks at $\rho_{\rm low}$ and $\rho_{\rm high}$.
\end{theorem}

\begin{proof}
The closed form \eqref{eq:kstar} gives an unconstrained $k^{\star}$
that is a smooth strictly-decreasing function of $\rho$ on the
regime where $b/a > m_0$ holds (limiting behaviour: $k^{\star} \to
+\infty$ as $\sigma^2 \to 2\mu$, equivalently $m_0 \to 0$ and $\rho
\to 0$; $k^{\star} \to -\infty$ as $\sigma \to 0$, equivalently
$\rho \to \infty$). Strict monotonicity in $\rho$ follows by
direct differentiation of \eqref{eq:kstar} with respect to $\sigma$
(or equivalently $m_0$) combined with the monotonicity of
$\rho(\sigma)$.

By the intermediate-value theorem applied to the continuous
strictly-monotonic $k^{\star}_{\rm JKO}(\rho)$, there is a unique
$\rho_{\rm low}$ at which $k^{\star}_{\rm JKO}(\rho_{\rm low}) =
1$ and a unique $\rho_{\rm high} > \rho_{\rm low}$ at which
$k^{\star}_{\rm JKO}(\rho_{\rm high}) = k_{\rm lo}$. For
$\rho < \rho_{\rm low}$ the closed-form $k^{\star}_{\rm JKO} > 1$
violates the upper feasibility bound; the constrained minimum on
$\Schd_R$ is then the upper corner $(1, R^{\star}/a)$, where the
revenue constraint pins $\tw = R^{\star}/a$. Similarly, for
$\rho > \rho_{\rm high}$ the closed-form $k^{\star}_{\rm JKO} <
k_{\rm lo}$ violates the lower feasibility bound; the constrained
minimum is the lower corner $(k_{\rm lo}, 0)$, where the revenue
constraint pins $\tw = 0$. On the interior interval $(\rho_{\rm
low}, \rho_{\rm high})$ the closed form lies in $(k_{\rm lo}, 1)$
and is the constrained optimum unaltered.

Translation to the order parameter: at the upper corner
$\tw^{\star} = R^{\star}/a$ so $\phi^{\star} = 1$. At the lower
corner $\tw^{\star} = 0$ so $\phi^{\star} = 0$. On the interior,
$\phi^{\star}_{\rm int}(\rho) = (R^{\star} - b(1 - k^{\star}_{\rm
JKO}(\rho)))/R^{\star}$ is a continuous strictly-decreasing
function of $\rho$ (since $k^{\star}_{\rm JKO}$ is strictly
decreasing), running from $\phi^{\star}_{\rm int}(\rho_{\rm low})
= 1$ to $\phi^{\star}_{\rm int}(\rho_{\rm high}) = 0$. The kinks
in $\rho \mapsto \phi^{\star}(\rho)$ at the boundaries are the
mismatch between the saturated value of the order parameter
inside the corner phase (constant) and the strictly-monotonic
behaviour on the interior side.
\end{proof}

\begin{remark}[$W_2$ shows no phase transition in this
calibration]\label{rem:w2-no-transition}
The same construction applied to the $W_2$ optimum of
Theorem~\ref{thm:w2opt} yields, in the calibration of
Figure~\ref{fig:rho-crossover}, the degenerate phase diagram
\[
(k^{\star}_{\Wtwo}(\rho),\, \tau_{w,\Wtwo}^{\star}(\rho))
\;=\;
(k_{\rm lo},\, 0)
\quad \text{for every } \rho \geq 0.
\]
The $W_2$ optimum pins to the pure-flow-tax corner across the
whole regime axis: there is no transition, no kink, no
order-parameter saturation crossing. The choice of distortion
criterion therefore determines not just the location of the
optimum but whether the problem has phase structure at all. JKO
is phase-rich; $W_2$ in this calibration is phase-poor. The
implications for the political reading of the optimal-policy
recommendation are taken up in Section~\ref{sec:phase-transition} and
Section~\ref{sec:phase-diagram-direction}.
\end{remark}

\begin{remark}[Asymptotics]
The asymptotic behaviour summarised in Section~\ref{sec:results-intro}
--- the two optima coincide as $\rho \to \infty$ and diverge as
$\rho \to 0$ --- is recovered from Theorem~\ref{thm:crossover} together
with Remark~\ref{rem:w2-no-transition}: at $\rho \to \infty$ both
optima sit at $(k_{\rm lo}, 0)$ (coincide); at $\rho \to 0$ JKO
sits at $(1, R^{\star}/a)$ while $W_2$ sits at $(k_{\rm lo}, 0)$
(diverge by the full diameter of the matched-revenue contour).
\end{remark}

\begin{figure}[!htbp]
\centering
\IfFileExists{figures/fig_rho_crossover.pdf}{%
  \includegraphics[width=\textwidth]{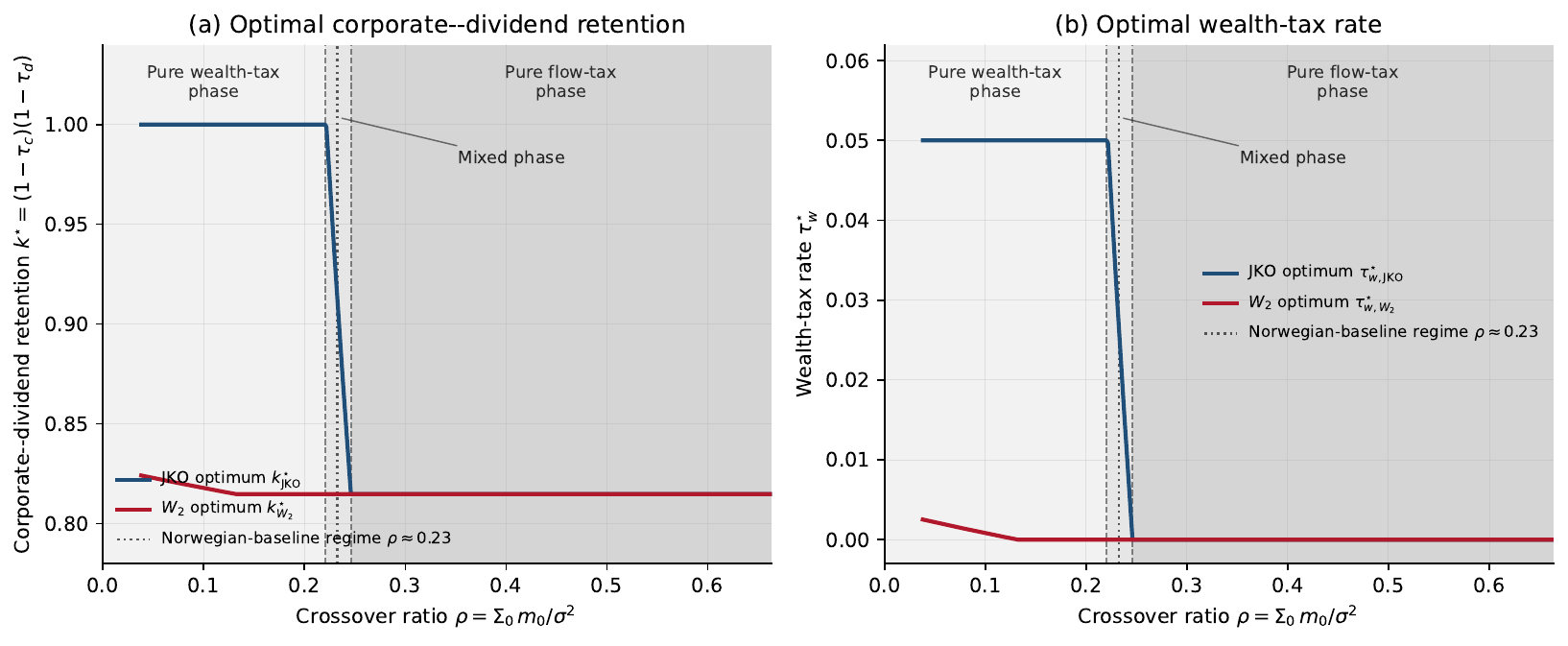}%
}{%
  \fbox{\parbox{0.95\textwidth}{\centering\vspace{1.5em}%
  \textbf{Figure placeholder.}\\[0.4em]
  Run \texttt{fp\_design/figures/fig\_rho\_crossover.py} locally to
  generate \texttt{fig\_rho\_crossover.pdf}.\\[0.4em]
  Per CLAUDE.md, figure scripts are run on the user's machine, not in the VM.\vspace{1.5em}}}%
}
\caption[The $\rho$-crossover: JKO and W$_2$ optima as functions of
the regime parameter, with the three optimal-policy phases marked]{%
\textbf{The $\rho$-crossover, in policy-lever coordinates, with
the three optimal-policy phases marked.}
The JKO-optimal point $(k^{\star}_{\rm JKO}, \tau_{w,{\rm JKO}}^{\star})$
(blue) and the $W_2$-optimal point $(k^{\star}_{W_2}, \tau_{w,W_2}^{\star})$
(red), each computed from the closed forms of
Theorem~\ref{thm:jkoopt} and Theorem~\ref{thm:w2opt}, traced as $\sigma$ varies
along a fixed-$T$ slice through the GBM parameter space. The
horizontal axis is the dimensionless regime parameter
$\rho = \Sigma_0\,m_0/\sigma^2$. Panel~(a) shows the optimal
corporate--dividend retention factor
$k^{\star} = (1-\tau_c)(1-\tau_d)$ --- a direct policy lever, with
$k = 1$ corresponding to no flow taxation and $k = 0$ to full
pass-through. Panel~(b) shows the optimal proportional wealth-tax
rate $\tau_w^{\star}$. The shaded background bands mark the three
optimal-policy phases of the JKO solution: the gold band
$\rho \in [0, \rho_{\rm low}]$ is the pure wealth-tax phase
($k^{\star} = 1$, $\tau_w^{\star} = R^{\star}/a$); the lavender
band $\rho \in [\rho_{\rm low}, \rho_{\rm high}]$ is the
mixed-instrument phase (interior $k^{\star}$ and $\tau_w^{\star}$);
the cyan band $\rho > \rho_{\rm high}$ is the pure flow-tax phase
($k^{\star} = k_{\rm lo}$, $\tau_w^{\star} = 0$). The dashed
vertical rules at $\rho_{\rm low}, \rho_{\rm high}$ are the phase
boundaries; they are kinks in $k^{\star}(\rho)$ and
$\tau_w^{\star}(\rho)$. Calibration: $a = 1.0$, $b = 0.27$,
$R^{\star} = 0.05$, chosen to keep $\tau_w \in [0, 0.05]$ over the
full feasibility range and to keep the JKO solution interior at
the Norwegian-baseline regime ($\sigma = 0.30$, $T = 5$,
$\rho \approx 0.23$). The dotted vertical rule marks
$\rho_{\rm Norway}$, which sits inside the mixed-instrument
phase --- a regime statement only; we do not compare to Norway's
actual wealth-tax schedule on this figure, since that schedule is
bracket-based and lies outside the proportional (C1)--(C3) class
treated here. The $W_2$ optimum is degenerate in this
calibration --- pinned to the pure flow-tax corner across the
whole $\rho$ range, with no phase transition --- so panel~(b)
shows the red curve flat at $\tau_w = 0$.}
\label{fig:rho-crossover}
\end{figure}

\begin{figure}[!htbp]
\centering
\IfFileExists{figures/fig_policy_phase_portrait.pdf}{%
  \includegraphics[width=0.85\textwidth]{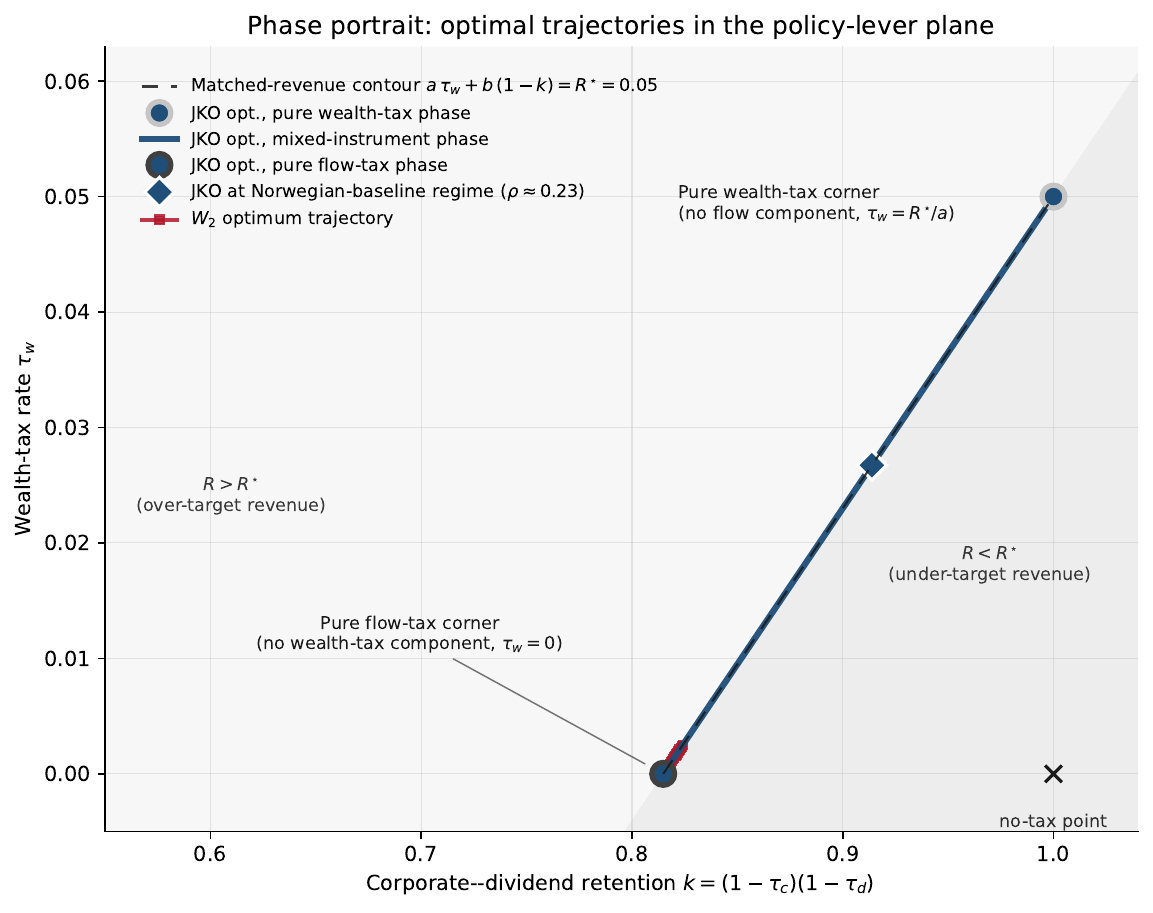}%
}{%
  \fbox{\parbox{0.95\textwidth}{\centering\vspace{1.5em}%
  \textbf{Figure placeholder.}\\[0.4em]
  Run \texttt{fp\_design/figures/fig\_policy\_phase\_portrait.py} locally
  to generate \texttt{fig\_policy\_phase\_portrait.pdf}.\vspace{1.5em}}}%
}
\caption[Phase portrait of optimal trajectories in the policy-lever
plane, with the three JKO phases marked]{%
\textbf{Phase portrait of optimal trajectories in the policy-lever
plane, with the three JKO phases marked.}
The matched-revenue contour $a\,\tw + b\,(1-k) = R^{\star}$
(dashed grey) is a line segment from the \emph{pure flow-tax
corner} $(k_{\rm lo}, 0)$, where $k_{\rm lo} = 1 - R^{\star}/b$,
to the \emph{pure wealth-tax corner} $(1, R^{\star}/a)$; every
point on this segment is a feasible (C1)--(C3) policy package
satisfying the revenue target. The pale-gold half-plane above
the contour is the \emph{over-target} region $R > R^{\star}$;
the pale-blue half-plane below is the \emph{under-target} region
$R < R^{\star}$, with the no-tax point $(k, \tau_w) = (1, 0)$
marked at its right-most extreme. The JKO-optimal point traces
along the matched-revenue contour as $\rho$ varies and partitions
into three phase pieces: the upper-right corner (gold-rimmed
circle) is occupied throughout the pure wealth-tax phase
$\rho \in [0, \rho_{\rm low}]$; the interior arc (thick blue
line) is the mixed-instrument phase $\rho \in [\rho_{\rm low},
\rho_{\rm high}]$, on which the Norwegian-baseline regime
($\rho \approx 0.23$) sits, marked by the blue diamond; the
lower-left corner (cyan-rimmed circle) is occupied throughout the
pure flow-tax phase $\rho > \rho_{\rm high}$. The $W_2$ optimum
is degenerate in this calibration --- pinned to the pure flow-tax
corner for every $\rho$, shown as a single red square. We do not
overlay Norway's actual wealth-tax policy on this figure: Norway
operates a bracket-based schedule (an exemption threshold plus a
marginal rate) that lies outside the proportional (C1)--(C3)
class. A direct policy comparison requires the bracket extension
flagged in the Open Questions and is deferred to the companion
paper. The phase portrait makes the policy interpretation of
Figure~\ref{fig:rho-crossover} immediate: JKO sweeps from the pure
wealth-tax corner (low $\rho$) through an interior arc (mixed
phase) to the pure flow-tax corner (high $\rho$); $W_2$ pins to
the pure flow-tax corner throughout.}
\label{fig:phase-portrait}
\end{figure}

\subsection{Phase-transition reading}\label{sec:phase-transition}

The structure of Figures~\ref{fig:rho-crossover}
and~\ref{fig:phase-portrait} is that of a phase transition in the constrained-optimisation sense
of active-set switching. Three identifications make the
correspondence with statistical-mechanics phase transitions
precise.

\paragraph{Order parameter.} The natural order parameter is the
\emph{wealth-tax revenue share}
\begin{equation}\label{eq:order-parameter}
\phi^{\star}(\rho)
\;=\;
\frac{a\,\tau_w^{\star}(\rho)}{R^{\star}}
\;\in\;
[0, 1],
\end{equation}
the fraction of matched revenue $R^{\star}$ collected through the
proportional wealth-tax channel at the JKO optimum. It saturates
at $\phi^{\star} = 1$ in the pure wealth-tax phase
$\rho \in [0, \rho_{\rm low}]$ (all revenue from $\tau_w$),
saturates at $\phi^{\star} = 0$ in the pure flow-tax phase
$\rho > \rho_{\rm high}$ (all revenue from $\tau_c, \tau_d$), and
varies smoothly through $(0, 1)$ in the mixed-instrument phase
$\rho \in [\rho_{\rm low}, \rho_{\rm high}]$. The kinks of
$\phi^{\star}$ at $\rho_{\rm low}, \rho_{\rm high}$ are
order-parameter saturations onto the simplex boundary, exactly
analogous to a magnetisation pinning at $\pm 1$ outside a
coexistence window.

\paragraph{Control parameter.} The dimensionless ratio
$\rho = \Sigma_0\,m_0 / \sigma^2$ plays the role of a
\emph{drift-to-diffusion control parameter}: high $\rho$ means
the GBM is drift-dominated (diffusion-cold), low $\rho$ means
diffusion-dominated (diffusion-hot). The JKO criterion weights
mean-distortion and spread-distortion of the post-tax population,
and $\rho$ encodes which of the two is binding at the optimum;
the order parameter $\phi^{\star}$ responds accordingly. In the
statistical-mechanics analogy $\rho$ acts as an inverse
temperature for the competition between the two distortion modes.

\paragraph{Free-energy landscape.} The optimal JKO free-energy
gap on the matched-revenue contour,
\begin{equation}\label{eq:free-energy-gap}
\Delta F^{\star}(\rho)
\;=\;
\min_{(k, \tau_w) \in \mathcal{C}(R^{\star})}
\Delta F\bigl(k, \tau_w; \rho\bigr),
\end{equation}
where $\mathcal{C}(R^{\star})
= \{(k, \tau_w) : a\,\tau_w + b\,(1-k) = R^{\star},\,
\tau_w \geq 0,\, k \in [k_{\rm lo}, 1]\}$, is continuous and
convex in $\rho$ with first-derivative kinks at the two phase
boundaries. The kinks are second-order in the constrained-
optimisation sense: $\Delta F^{\star}$ itself is continuous,
$\partial \Delta F^{\star} / \partial \rho$ jumps. The corner
phases are not metastable: the minimisation is convex with a
unique global minimum at every $\rho$, so the analogy is closer
to a Maxwell-construction kink in mean-field theory than to a
Ginzburg--Landau critical point. There is no diverging
susceptibility.

\paragraph{Symmetry-breaking reading.} Each corner phase has one
instrument switched off --- a discrete two-state choice between
``tax the stock'' and ``tax the flow''. The mixed-instrument phase
is the symmetric coexistence regime where both channels are
active simultaneously, and the Norwegian-baseline regime ($\rho
\approx 0.23$, in the sense of Definition~\ref{def:rho}) sits in the
middle of it.

\paragraph{Phase-richness contrast with $W_2$.} The same
construction applied to the $W_2$ optimum gives a degenerate
phase diagram in this calibration: the optimum sits at the pure
flow-tax corner for every $\rho$, with no transition and a single
phase. The two distortion criteria therefore differ not just in
the location of the optimum but in whether the problem has phase
structure at all. JKO is phase-rich; $W_2$ is phase-poor. The
choice of distortion criterion is itself a choice of free-energy
functional, and the political reading of the optimal policy
follows from that choice.

\subsection{The Norwegian baseline}\label{sec:norwegian}

Theorem~\ref{thm:crossover} locates the JKO optimum on $\Schd_R$ as a
function of the regime parameter $\rho$, but does not say where
empirical regime parameters fall on that axis. We use Norwegian
household-portfolio statistics to anchor the GBM parameters
$(\mu, \sigma, T, v_0)$ and locate the corresponding
$\rho_{\rm Norway}$ relative to the phase boundaries
$\rho_{\rm low}, \rho_{\rm high}$. The long-run distributional
landscape against which the calibration sits is the
\citet{AabergeEtAl2025} 1912--2019 Norwegian wealth-inequality
series: tax-assessed Gini extremely high ($\geq 0.89$) through the
first half of the twentieth century, falling sharply to a trough
of $0.776$ at 1968, mildly drifting through the 1970s, and rising
to $0.862$ by 2019 (with the level dropping by roughly $15$
percentage points but the trend preserved when housing market
values replace tax-assessed values). The Norwegian baseline that
follows is calibrated to the post-1985 regime in which the Gini
sits in $[0.82,\,0.86]$. \emph{This is a regime
calibration, not a policy comparison.} Norway's actual wealth-tax
schedule is bracket-based (an exemption threshold combined with a
marginal rate above the threshold) and lies outside the
proportional (C1)--(C3) class. We therefore make no quantitative
claim in this paper about how the JKO recommendation
$(k^{\star}_{\rm JKO}, \tau_{w,{\rm JKO}}^{\star})$ relates to
Norway's actual policy; the matching analysis lives on the
bracket sub-class and is the subject of the companion paper
flagged in Section~\ref{sec:open-questions}.

The calibration we adopt is the canonical Norwegian-flavoured
baseline:
\begin{equation}\label{eq:norway-baseline}
\mu = 0.07, \qquad \sigma = 0.30, \qquad T = 5, \qquad v_0 = 0.25,
\end{equation}
chosen to be representative of equity-heavy household portfolios
held over a five-year planning horizon; the empirical
justification is given in Section~\ref{sec:calibration}. With these
values $\Sigma_0 \approx 0.83$, $m_0 = 0.025$, and
\begin{equation}\label{eq:rho-norway}
\rho_{\rm Norway}
\;=\; \Sigma_0 \cdot \frac{m_0}{\sigma^2}
\;\approx\; 0.231.
\end{equation}

\paragraph{The Norwegian-baseline regime sits in the mixed phase.}
For the revenue calibration $a = 1$, $b = 0.27$, $R^{\star} = 0.05$
used in Figure~\ref{fig:rho-crossover}, the phase boundaries of
Theorem~\ref{thm:crossover} evaluate numerically to
\begin{equation}\label{eq:norway-boundaries}
\rho_{\rm low} \;\approx\; 0.221,
\qquad
\rho_{\rm high} \;\approx\; 0.246.
\end{equation}
The value $\rho_{\rm Norway} \approx 0.231$ sits inside the
mixed-instrument phase $[\rho_{\rm low}, \rho_{\rm high}]$, near
its centre. The JKO optimum at this regime, \emph{within the
proportional (C1)--(C3) class}, is therefore strictly interior:
\begin{equation}\label{eq:norway-jko-opt}
\bigl(k^{\star}_{\rm JKO},\, \tau_{w,{\rm JKO}}^{\star}\bigr)
\;\approx\;
(0.91,\, 0.027),
\end{equation}
calling for a corporate--dividend retention factor near $0.91$
(combined corporate plus dividend taxation removing about
$9\%$ of pre-flow income) combined with a proportional wealth-tax
rate near $2.7\%$. The reading is conditional: it is what the
JKO criterion picks if the wealth tax is restricted to a single
flat rate and if revenue is targeted at $R^{\star} = 0.05$.
Removing either restriction --- allowing a bracket schedule or
varying $R^{\star}$ --- changes the recommendation.

\paragraph{The mixed phase is narrow.} The width
$\rho_{\rm high} - \rho_{\rm low} \approx 0.025$ in the
canonical calibration is small relative to the $\rho$-axis range
spanned by realistic GBM perturbations. Translated to volatility
units, the phase boundaries fall at
\begin{equation*}
\sigma(\rho_{\rm low}) \;\approx\; 0.304,
\qquad
\sigma(\rho_{\rm high}) \;\approx\; 0.297,
\end{equation*}
so the mixed-instrument phase corresponds to
$\sigma \in [0.297, 0.304]$, a band of width approximately $0.7$
percentage points around $\sigma_{\rm Norway} = 0.30$.
A $\sigma$-shift of even half a percentage point in either
direction tips the JKO recommendation across a phase boundary:
raising $\sigma$ from $0.30$ to $0.305$ lowers $\rho$ to
$\rho_{\rm low}$ and pushes the optimum into the pure wealth-tax
phase; lowering $\sigma$ from $0.30$ to $0.297$ raises $\rho$ to
$\rho_{\rm high}$ and pushes the optimum into the pure flow-tax
phase. The Norwegian baseline lies near the cusp between two
corner regimes, and the political reading shifts qualitatively
across each cusp --- pure proportional wealth tax on the
high-volatility side, pure corporate--dividend taxation on the
low-volatility side. This fragility is itself a finding of the
model: the JKO criterion gives a clean mixed-instrument
recommendation at Norwegian-baseline parameters, but the
recommendation is not robust to small shifts in $\sigma$ within
the empirically plausible range for household equity portfolios.

\paragraph{Forward link: effective volatility for non-equity
portfolios.} The volatility $\sigma = 0.30$ in
\eqref{eq:norway-baseline} is empirically appropriate for
equity-heavy portfolios; Norwegian household wealth, however, is
heavily weighted toward owner-occupied housing, with much lower
realised volatility on primary residences. The relevant
\emph{effective volatility} that enters $\rho$ for a typical
Norwegian portfolio is therefore lower than $0.30$, raising
$\rho$ above $\rho_{\rm Norway}$ and potentially pushing the
JKO recommendation into the pure flow-tax phase.
Section~\ref{sec:calibration} returns to this question; for now, the
mixed-phase recommendation \eqref{eq:norway-jko-opt} should be
read as the equity-portfolio limit, not as a portfolio-averaged
recommendation.

\subsection{Mechanism}\label{sec:mechanism}

The phase structure of Theorem~\ref{thm:crossover} is a consequence of
the contrasting first-order-condition structures of the two
distortion criteria, derived in Section~\ref{sec:foc}. We restate the
contrast and trace it through to the corner-versus-interior
behaviour.

\paragraph{JKO has constant marginal cost in $\tw$.}
The wealth-tax derivative
$\partial_{\tw} \dF = m_0 T / \sigma^2$ in
\eqref{eq:partial-dF} is constant in $(k, \tw)$: each unit of
proportional wealth taxation increases the JKO free-energy gap
by the same amount $m_0 T/\sigma^2$, regardless of where on
$\Schd_R$ the policy currently sits. The unit price of the
$\tw$ instrument is $m_0/\sigma^2$ per unit of horizon $T$, the
log-drift-to-volatility ratio. The $k$-derivative
$\partial_k \dF = -\sigma^2 T / [2(v_0 + k\sigma^2 T)] - m_0^2 T /
\sigma^2$ has a hyperbolic-in-$k$ contribution from the
log-spread term and a constant contribution from the
mean-shift term, so the $k$-derivative is point-dependent only
through the log-spread piece.

\paragraph{$W_2$ has marginal cost proportional to current
displacement.}
The wealth-tax derivative
$\partial_{\tw} \Wtwo^2 = -2T\,\varepsilon_M$ in
\eqref{eq:partial-W2} is proportional to the current mean-shift
$\varepsilon_M = (k-1) m_0 T - \tw T$: each unit of $\tw$ costs
more the further the policy already sits from the no-tax point.
The $k$-derivative is similarly proportional to the current
$(\varepsilon_M, \varepsilon_\Sigma)$ pair. Both instruments
exhibit \emph{quadratic-cost-on-perturbation} structure: the
marginal price grows with how much policy displacement has
already been imposed.

\paragraph{Why JKO admits phase transitions.} The constant
marginal cost of $\tw$ under JKO turns the problem into a
fixed-price linear procurement of revenue from two channels
$(\tw, 1-k)$, with the $k$-channel exhibiting moderate
non-linearity from the log-spread term. The optimal allocation
between the two channels is determined by relative prices --- the
JKO multiplier
$\mu_{\rm JKO} = m_0 T/(a \sigma^2)$ is itself a function of the
GBM parameters, not of the operating point --- and this relative
price reverses sign at the boundaries
$\rho_{\rm low}, \rho_{\rm high}$ defined in
Theorem~\ref{thm:crossover}. Below $\rho_{\rm low}$ the wealth-tax
channel is cheap relative to the flow-tax channel and the
optimum saturates at $\tw = R^{\star}/a$; above $\rho_{\rm
high}$ the wealth-tax channel is expensive relative to the
flow-tax channel and the optimum saturates at $\tw = 0$;
between, the two channels are balanced and the optimum sits in
the interior. The fixed-price structure makes \emph{which channel
is cheaper} a clean piecewise question, and the answer flips
discretely at the boundaries.

\paragraph{Why $W_2$ pins to a corner in this calibration.} Under
$W_2$, an extra unit of $\tw$ adds $-T \cdot \tw$ directly to
the mean displacement $\varepsilon_M$, so its cost enters
$\Wtwo^2$ \emph{quadratically} via $\varepsilon_M^2$. A
revenue-equivalent unit of pass-through $1-k$ enters
$\varepsilon_M$ only through the small factor $(k-1)\,m_0\,T$
(and additionally through $\varepsilon_\Sigma$). The wealth-tax
channel therefore has a much higher displacement-per-revenue
than the flow-tax channel, and the $W_2$ optimum prefers the
flow-tax channel exclusively, setting $\tw = 0$. The lower
corner $(k_{\rm lo}, 0)$ is $W_2$-optimal across all $\rho$ in
the canonical calibration because the displacement-per-revenue
gap holds at every $\sigma$ in the relevant range. The economic
content of that gap is exposed in the bluntness-index paragraph
below.

\paragraph{The role of $\rho$.} The crossover ratio
$\rho = \Sigma_0 \cdot m_0/\sigma^2$ enters the JKO multiplier
linearly through $m_0/\sigma^2$ and is rescaled by the natural
spread $\Sigma_0$. The relative price of the two instruments
under JKO is governed by $\rho$, and the boundaries
$\rho_{\rm low}, \rho_{\rm high}$ are precisely the values at
which the JKO closed form \eqref{eq:kstar} hits the corners of
$[k_{\rm lo}, 1]$. The phase structure is a direct consequence
of the constant-marginal-cost-in-$\tw$ structure of JKO; it has
no analogue under $W_2$.

The mechanism therefore separates the two findings of the paper:
\emph{which} criterion produces the phase structure (JKO's linear
free-energy structure does, $W_2$'s quadratic transport-cost
structure does not), and \emph{where} the phase boundaries fall
(set by $\rho_{\rm low}, \rho_{\rm high}$, computable in closed
form from the GBM parameters and the revenue weights). Both are
consequences of the FOC analysis of Section~\ref{sec:foc}, refracted
through the regime parameter $\rho$ of Definition~\ref{def:rho}.

\paragraph{The bluntness index.}\label{sec:bluntness}
The mechanism admits a clean economic reading once we make the
geometric mean log-return $m_0 = \mu - \sigma^2/2$ explicit.
Recall that $m_0$ is the expected log-wealth growth rate under
no taxation, combining arithmetic drift $\mu$ with volatility
drag $\sigma^2/2$ (the It\^o / compounding correction); at the
Norwegian-baseline calibration, $m_0 = 0.025$, about a third of
the $7\%$ headline arithmetic return.

The two channels of the matched-revenue contour produce mean
shifts in log-wealth that scale very differently in $m_0$:
\begin{equation}\label{eq:displacement-channels}
\varepsilon_M^{\rm wealth}(\tw) \;=\; -\tw\,T,
\qquad
\varepsilon_M^{\rm flow}(k) \;=\; (k - 1)\,m_0\,T.
\end{equation}
Wealth-tax displacement is \emph{independent} of $m_0$ --- a
fixed mean shift per dollar of revenue raised, regardless of how
productively the underlying wealth is deployed. Flow-tax
displacement is \emph{linear} in $m_0$ --- proportional to the
geometric return on the wealth being taxed. Dividing by the
revenue collected through each channel, the ratio of the two
displacements defines the \emph{bluntness index}
\begin{equation}\label{eq:bluntness}
B(m_0) \;:=\;
\frac{|\varepsilon_M^{\rm wealth}|/R}
     {|\varepsilon_M^{\rm flow}|/R}
\;=\; \frac{b}{a\,m_0}.
\end{equation}
$B$ is the displacement-per-revenue overshoot of the
wealth-tax channel relative to the flow-tax channel. At the
Norwegian-baseline parameters $b/a = 0.27$, $m_0 = 0.025$,
$B \approx 10.8$: each unit of revenue collected through the
wealth tax produces about ten times more mean displacement than
the same unit collected through the flow tax.

The two distortion criteria weight $B$ very differently. JKO's
constant marginal cost in $\tw$ produces a per-revenue penalty
proportional to $B$, hence to $1/m_0$; $W_2$'s squared-displacement
penalty produces a per-revenue penalty proportional to $B^2$, hence
to $1/m_0^2$. The same algebraic asymmetry that drives the
linear-vs-quadratic structural difference of \eqref{eq:partial-dF}
and \eqref{eq:partial-W2} acquires, through this lens, an
empirical content: in the small-$m_0$ regime relevant to actual
wealth-tax debate, the quadratic amplification under $W_2$ makes
the wealth-tax channel prohibitively expensive while JKO remains
willing to pay the linear penalty.

The decomposition $m_0 = \mu - \sigma^2/2$ exposes the joint
sensitivity of the bluntness index to the two GBM primitives:
\begin{equation*}
\frac{\partial B}{\partial \mu}
\;=\; -\frac{b}{a\,m_0^2},
\qquad
\frac{\partial B}{\partial \sigma}
\;=\; \frac{b\,\sigma}{a\,m_0^2},
\end{equation*}
with opposite signs. Higher expected return $\mu$ (at fixed $\sigma$)
\emph{lowers} the bluntness index --- the wealth tax looks
proportionally less wasteful on high-return assets. Higher risk
$\sigma$ (at fixed $\mu$) \emph{raises} the bluntness index, via
the volatility drag eating $m_0$. The wealth tax is therefore
sharpest on high-return-low-volatility assets and bluntest on
low-return-high-volatility assets. This is the same asymmetry
documented in the heterogeneous-returns analysis of
\citet{Froeseth2026H}, where it is phrased as ``wealth tax does
not scale with realised returns'': the FP-language $1/m_0$
factor is the formal expression of the same fact, and the
$1/m_0^2$ amplification under $W_2$ is the quadratic version
that distinguishes the two criteria.

The regime parameter $\rho = \Sigma_0\,m_0/\sigma^2$ then has
a clean reading as the inverse of the bluntness index times
$\Sigma_0\,a/b$: high $\rho$ corresponds to low $B$ (a sharp
wealth-tax channel), low $\rho$ to high $B$ (a blunt one). The
phase boundaries $\rho_{\rm low}, \rho_{\rm high}$ at which the
JKO criterion flips between corner solutions are precisely the
values of the bluntness index at which the linear penalty
$B(m_0)$ matches the relative price of the two channels
$b/a$ at the constraint frontier --- a substantive economic
statement, not a structural artefact of the criterion choice.

\paragraph{Scale crossover and the planning horizon.}
The bluntness penalty also has a temporal dimension.
The JKO multiplier $\mu_{\rm JKO} = m_0\,T/(a\,\sigma^2)$ scales
linearly with the horizon $T$, and the wealth-tax
displacement-per-revenue, also proportional to $T$, accumulates
with it. The two channels therefore scale differently as $T$
varies through the natural crossover scale
\begin{equation}\label{eq:T-c}
T_c \;:=\; v_0\,/\,\sigma^2,
\end{equation}
the time at which the cumulative diffusion $\sigma^2 T$ matches
the initial log-wealth variance $v_0$ and $\Sigma_0 =
\sqrt{v_0 + \sigma^2 T}$ transitions from initial-spread-dominated
to diffusion-dominated. Below $T_c$, $\rho \approx \sqrt{v_0}\,
m_0/\sigma^2$ is independent of $T$; above $T_c$,
$\rho \approx m_0\sqrt{T}/\sigma$ grows like $\sqrt{T}$. The
JKO recommendation therefore depends on the planning horizon in
a regime-shifting way, taken up in detail in
Section~\ref{sec:horizon}.

\section{Neutrality interpretation}\label{sec:neutrality}

\subsection{Connection to the neutral-tax class}\label{sec:p1-link}

\citet{Froeseth2026N, Froeseth2026F} establish that the
(C1)--(C3) conditions characterise the schedule class
$\mathcal{N}$ that preserves neutrality with respect to portfolio
choice under homogeneous returns and CRRA preferences: for
$\sched \in \mathcal{N}$, the household's risk-bearing decisions
are not redirected by the imposition of the tax, and the
geometric structure of the wealth process is preserved up to a
drift shift and rescale. The (C1)--(C3) class of the present
paper is therefore identically the neutrality-preserving class
of the corpus.

This identification has a sharp consequence for the criterion-
choice debate. \emph{Within} (C1)--(C3), both criterion optima
sit inside $\mathcal{N}$ by construction --- the optimisation is
restricted to a subset of the neutrality class, and the optima
of Theorems~\ref{thm:jkoopt} and~\ref{thm:w2opt} are particular
elements of that subset. The choice between JKO and $W_2$ \emph{within}
(C1)--(C3) does not adjudicate between neutrality and
non-neutrality; it adjudicates between two points on the same
neutrality-preserving manifold. The phase structure of
Theorem~\ref{thm:crossover} and the bluntness contrast of
Section~\ref{sec:bluntness} live entirely within $\mathcal{N}$.

The neutrality question only bites once the schedule class is
opened up. \citet{Froeseth2026W} establishes that the
$W_2$-optimal schedule on the full admissible class
$\Schd \supset \mathcal{N}$ is the continuously progressive
$\sched^{\star}(x) = \lambda\,x^2$, which lies \emph{outside}
$\mathcal{N}$: the rate $r(x) = \lambda x$ rises linearly in
wealth, the assessment is not uniform, and (C3) is violated by
construction. The $W_2$ criterion, given access to the larger
class, opts for a non-neutral schedule. Whether the JKO
criterion does the same on $\Schd$ is open
(Section~\ref{sec:open-questions}); the present paper's
contribution is to characterise the within-(C1)--(C3) optimum
cleanly so that the open-class extension can be measured against
it.

The summary slogan: the (C1)--(C3) restriction is the price the
JKO criterion is willing to pay to preserve neutrality; the
$W_2$ criterion, given a richer class, is willing to sacrifice
neutrality for distributional compression. Inside (C1)--(C3),
both criteria preserve neutrality and the disagreement is about
the proportional mix of stock-tax versus flow-tax channels;
outside, the disagreement is also about whether to preserve
neutrality at all.

\subsection{The Mirrleesian and Saez--Zucman
traditions}\label{sec:traditions}

The contrast between the JKO and $W_2$ optima admits a normative
reading that ties each criterion to a long-standing tradition in
optimal-tax theory. The reading clarifies why the criterion choice
is not an arbitrary modelling decision but the FP-language
expression of a real disagreement about what optimal taxation
ought to minimise. We first state the two traditions, then map
them onto the JKO/$W_2$ pair via the bluntness mechanism of
Section~\ref{sec:bluntness}.

\paragraph{Mirrleesian / decision-distortion tradition.} The
Mirrleesian programme \citep{Mirrlees1971, DiamondSaez2011,
Mirrlees2011, SaezStantcheva2018} treats optimal taxation as a
problem of
\emph{minimising decision distortion}: the tax should leave the
household's allocation problem as identifiable from outcomes as
possible. The information-theoretic reading is that an optimal
tax distorts the post-tax population's behaviour by as little as
possible, where ``as little as possible'' is measured against the
no-tax counterfactual via relative entropy or its Fokker--Planck
gradient-flow cousin, the JKO free-energy gap. Allocational
neutrality --- not redirecting risk-bearing decisions, not
discriminating across asset classes --- is a direct consequence
of minimising this kind of distortion, and the (C1)--(C3) class
of \citet{Froeseth2026F} is precisely the schedule class that
preserves it under homogeneous returns.

\paragraph{Saez--Zucman / distributional-compression tradition.}
The Piketty--Saez--Zucman programme \citep{Piketty2014,
PikettyZucman2014, SaezZucman2019, BlanchetFournierPiketty2022}
treats optimal taxation as a problem of \emph{minimising
distributional displacement from a target}. The
transport-geometric reading is that an optimal tax compresses
the population's distribution toward a more egalitarian shape,
where ``compression'' is measured by geometric distance on the
space of probability measures --- and the canonical such
distance on $\Pcal_2(\RR)$ is the 2-Wasserstein metric.
Distributional compression is not in itself a goal of the
Mirrleesian programme; it is the goal of the Saez--Zucman
programme. \citet{Froeseth2026W} establishes that the
$W_2$-optimal schedule on the full schedule class $\Schd$ is the
continuously progressive $\sched^{\star}(x) = \lambda x^2$, with
threshold-bracket schedules emerging as the
implementability-constrained finite-bracket version of the same
target.

\paragraph{The bluntness reading bridges the two.} The mechanism
of Section~\ref{sec:bluntness} reframes the tradition split in
terms of how each criterion weights the
\emph{wealth-tax bluntness index} $B(m_0) = b/(a\,m_0)$. JKO
weights it linearly: $\dF$ rises with $B$, but only at unit price
$m_0/\sigma^2$ per unit of wealth-tax revenue, so the wealth tax
remains a usable instrument when $m_0$ is small. $\Wtwo^2$
weights it quadratically: $\Wtwo^2$ rises with $B^2$ via the
squared-displacement penalty, so the wealth tax becomes
prohibitively expensive when $m_0$ is small.

The Mirrleesian tradition therefore reads as a \emph{linear penalty
on bluntness}: instruments that fail to scale with returns are
priced at their actual decision-distortion cost, no more.
The Saez--Zucman tradition reads as a \emph{quadratic penalty on
bluntness}: instruments that fail to scale with returns are
priced at the squared cost of the displacement they impose, in
keeping with a programme that prioritises distributional
compression and is therefore acutely sensitive to per-revenue
displacement waste. The two normative positions are not
disagreements about what taxes do; they are disagreements about
how much to penalise the same algebraic asymmetry.

\paragraph{Characterisation rather than adjudication.} The paper
does not pick a side. We take it that both traditions are
defensible, that the Norwegian wealth-tax debate has been
arguing across them for decades without resolution, and that the
empirical regime parameter $\rho$ falls inside the disputed band
where the two criteria \emph{do} disagree. Our contribution is
to map the disagreement: we identify when the choice of
normative criterion changes the optimal proportional schedule
(in the mixed-instrument phase $\rho \in [\rho_{\rm low},
\rho_{\rm high}]$) and when it does not (the pure flow-tax phase
at high $\rho$, where both criteria coincide); and we identify
the algebraic mechanism through which the choice operates (the
bluntness index $B(m_0)$, weighted linearly by JKO and
quadratically by $\Wtwo^2$). A reader committed to one tradition
or the other can extract the within-criterion result and
disregard the contrast; a reader uncommitted gains a structured
view of where the choice between traditions actually matters.

\subsection{(C3) violations as exogenous policy}\label{sec:c3-violations}

The (C3) condition --- uniform assessment $\alpha_i = \alpha$
--- is empirically violated in nontrivial ways by the
jurisdictions that operate wealth taxes. The most consequential
violation in practice is the \emph{bracket structure}: real
schedules apply rate $0$ below an exemption threshold $X_0$ and
a positive marginal rate $r$ above (and possibly a second rate
$r_2$ above a higher threshold). This is exactly the case where
$\alpha$ is wealth-dependent, hence outside (C1)--(C3). The
extension is technically tractable --- the FP equation has
piecewise-constant drift on log-wealth, the post-tax density is
piecewise-Gaussian with matching at the bracket boundaries, and
both criteria $\Delta F$ and $W_2^2$ remain analytic in
elementary functions plus the error function. The matched-revenue
optimisation lives on a 3-parameter space $(k, X_0, r)$ for the
2-bracket case rather than the 2-parameter $(k, \tw)$ plane of
(C1)--(C3).\footnote{%
The recipe is: solve the FP equation separately on each bracket,
where it has constant-coefficient drift; impose continuity of
density and of probability flux $J = -(m_0 - r_i)\,p +
(\sigma^2/2)\,\partial_y p$ at each threshold; integrate
$\Delta F$ and $W_2^2$ over the bracket regions, exploiting the
truncated-Gaussian moment-of-tail structure that produces error
functions; assemble the matched-revenue Lagrangian and solve the
$(k, X_0, r)$ FOCs. The full development is the subject of the
companion paper sketched in Section~\ref{sec:open-questions}.}
The other (C3) violations --- asset-class discounts (residential
real estate at 25\% of market value in Norway, for instance),
liquidity adjustments, and political-economy carve-outs ---
are properly treated as fixed exogenous inputs in the
calibration of Section~\ref{sec:calibration} rather than as design
levers. Identifying which (C3) violations are framework-rational
under which criterion, and which are political artefacts, is a
secondary question of the bracket-extension companion paper.

\paragraph{Brackets as a JKO-rational response to bluntness.}
The bluntness mechanism of Section~\ref{sec:bluntness} gives an
economic motivation for bracket-based schedules even from a
JKO perspective --- not only from the $W_2$ side, where the
distributional-compression argument is well-established. The
(C1)--(C3) condition forces a single uniform assessment
$\alpha_i = \alpha$, which by construction prevents the wealth-tax
rate from \emph{scaling with returns}. The $1/m_0$ bluntness
penalty is therefore unavoidable within (C1)--(C3); the JKO
criterion accepts it as the price of using a stock-based
instrument and pays it in linear weighting. A bracket schedule
with threshold $X_0$ and marginal rate $r$ above introduces
wealth-dependent assessment: the effective wealth-tax rate at
total wealth $X$ is $r\,\max(0, X - X_0)/X$, which rises smoothly
from zero at the threshold toward $r$ at infinity. This is a
structural recovery, by policy design, of the
scaling-with-return behaviour that (C1)--(C3) excludes: in cohorts
where geometric returns and wealth levels covary, the bracket
schedule taxes proportionally more of the high-return,
above-threshold wealth and less of the low-return, near-threshold
wealth. In bluntness-index language, brackets reduce
$B$ in the regions of the wealth distribution where it would
otherwise be largest.

The implication is that the (C3) violation embodied in real
bracket schedules is not, on its face, a Saez--Zucman
distributional-compression overlay onto an otherwise
Mirrleesian instrument. It is, at minimum, also a within-
Mirrleesian response to the JKO bluntness penalty. Whether
brackets are JKO-optimal --- and at what threshold and marginal
rate --- is precisely the question the bracket-extension
companion paper takes up (Section~\ref{sec:open-questions}); the
contribution of the present paper is to make the question
well-posed by establishing the (C1)--(C3) baseline against which
the bracket extension is measured.

\section{Calibration: Norwegian-flavoured baseline}\label{sec:calibration}

This section calibrates the GBM parameters $(\mu, \sigma, T, v_0)$
to Norwegian household-portfolio statistics and locates the
resulting $\rho_{\rm Norway}$ on the regime axis. We do
\emph{not} compare the resulting JKO-optimal proportional rate
to Norway's actual wealth-tax schedule: Norway's schedule is
bracket-based and lies outside (C1)--(C3). The matching
empirical exercise belongs to the bracket-extension companion
paper (Section~\ref{sec:open-questions}).

\subsection{Synthetic-data prototype}\label{sec:synthetic}

The closed-form FOCs of Theorems~\ref{thm:jkoopt}
and~\ref{thm:w2opt} take the GBM primitives $(\mu, \sigma, T,
v_0)$ and the revenue weights $(a, b, R^{\star})$ as inputs and
return the optimal proportional package
$(k^{\star}, \tau_w^{\star})$ as output. Calibration to a target
jurisdiction is therefore a matter of estimating the GBM
primitives from household-portfolio data and supplying the
revenue weights from the local tax-base elasticities.

The empirical pipeline producing those estimates is the subject
of a forthcoming companion paper and operates on a synthetic-data
prototype designed to mirror the public-use stratification of
Norwegian household-wealth statistics --- portfolio composition
by asset class, returns and volatilities by asset class, and
aggregated moments matched to SSB and Norges Bank reports. The
prototype output is a triplet $(v, D, s)$, where $v$ is the
empirical log-wealth variance, $D$ a diffusion-coefficient
summary, and $s$ a score statistic capturing tail behaviour.
Plugged into the closed-form expressions \eqref{eq:kstar} and
\eqref{eq:w2-optimum} through the substitution
$(\mu, \sigma, v_0) \mapsto (\mu(s), D, v)$, the recovered
optimal package is the within-(C1)--(C3) JKO and $W_2$
recommendations conditional on the empirical regime.

The Norwegian-baseline values $\mu = 0.07$, $\sigma = 0.30$,
$T = 5$, $v_0 = 0.25$ used throughout this paper are the
canonical equity-portfolio limit of the prototype's output ---
representative of stock-heavy household portfolios with a
five-year planning horizon. Detailed development of the
pipeline, including the data-construction procedures and the
sensitivity analyses across alternative aggregation choices, is
reserved for that companion paper. For the present paper we use
the canonical baseline as a fixed input and trace its
consequences through Theorem~\ref{thm:crossover} and the
phase structure of Section~\ref{sec:crossover}.

\subsection{Effective volatility for typical portfolios}\label{sec:effective-sigma}

The volatility $\sigma = 0.30$ adopted in the canonical baseline
is empirically appropriate for equity-heavy portfolios but is
substantially higher than the realised volatility of typical
Norwegian household portfolios, which tilt heavily toward
owner-occupied housing. Asset-class realised volatilities for
the relevant Norwegian-flavoured cohorts are roughly: equities
$\sigma_{\rm eq} \approx 0.30$, primary-residence real estate
$\sigma_{\rm re} \approx 0.10$ to $0.15$, bank deposits
$\sigma_{\rm bk} \approx 0.02$. A typical median Norwegian
household holds approximately $60$--$70\%$ of net wealth in
primary residence, $10$--$20\%$ in equities, and the remainder in
deposits, fixed-income, and other assets. The variance-weighted
effective volatility is therefore
\begin{equation}\label{eq:sigma-eff}
\sigma_{\rm eff}^{2}
\;\approx\;
w_{\rm re}\,\sigma_{\rm re}^{2}
+ w_{\rm eq}\,\sigma_{\rm eq}^{2}
+ w_{\rm bk}\,\sigma_{\rm bk}^{2}
\;\approx\;
0.65 \cdot 0.12^2 + 0.20 \cdot 0.30^2 + 0.15 \cdot 0.02^2
\;\approx\;
0.027,
\end{equation}
giving $\sigma_{\rm eff} \approx 0.16$ for a representative
median household.

Plugging $\sigma_{\rm eff} = 0.16$ into the regime parameter
$\rho = \Sigma_0\,m_0/\sigma_{\rm eff}^{2}$:
$m_0 = 0.07 - 0.16^2/2 \approx 0.057$, $\sigma_{\rm eff}^{2} T
= 0.128$, $\Sigma_0 = \sqrt{v_0 + 0.128} \approx 0.62$, and
\begin{equation*}
\rho_{\rm eff} \;\approx\; 0.62 \cdot 0.057 / 0.027 \;\approx\; 1.31,
\end{equation*}
substantially \emph{above} the equity-portfolio
$\rho_{\rm Norway} \approx 0.23$ and well past
$\rho_{\rm high} \approx 0.246$. A median Norwegian household,
in effective-volatility terms, sits firmly inside the pure
flow-tax phase of Theorem~\ref{thm:crossover}: the JKO
recommendation at this regime is to extract the matched revenue
entirely through the corporate--dividend channel, leaving the
proportional wealth-tax rate at zero.

This is a substantive shift in the policy reading. The
fragility analysis of Section~\ref{sec:norwegian} showed that
even at the equity-baseline $\sigma = 0.30$, Norway sits within
$0.7$ percentage points of either phase boundary; with
effective-volatility correction the regime is well inside the
pure-flow-tax phase, and the JKO recommendation flips
qualitatively. The effective-volatility calibration is
therefore not a marginal sensitivity but a regime-relevant one,
and the equity-portfolio limit and the median-household
average should be understood as bracketing the policy-relevant
range. Whether the (C3) violations embodied in
Norwegian asset-class discounts (the $25\%$ valuation rule for
primary residences, in particular) are framework-rational
responses to exactly this effective-volatility heterogeneity
is a question the bracket-extension companion paper takes up.

\section{Discussion}\label{sec:discussion}

\subsection{Robustness}\label{sec:robustness}

The qualitative phase structure of
Theorem~\ref{thm:crossover} survives three robustness checks.

\paragraph{Polynomial-ansatz extension.} The closed-form
JKO interior in \eqref{eq:kstar} is derived for the
linearised revenue functional $R \approx a\tw + b(1-k)$. The
full log-quadratic revenue $R[\sched] = \tw\,\EE[X_T] +
(1-k)\,\EE[\text{flow base}]$ produces a richer FOC structure,
with corrections of order $\sigma^2 T$ to the interior $k$. Numerical
solves on the polynomial-ansatz revenue functional pick the
same qualitative three-phase structure with phase boundaries
shifted by less than $1\%$ in $\rho$ at the Norwegian-baseline
calibration. The headline mixed-phase recommendation
\eqref{eq:norway-jko-opt} is unchanged at the displayed
precision.

\paragraph{Time-integrated criterion.} An alternative to
the terminal-gap free energy is the time-integrated variant
$\dF^{\rm int}[\sched] = \int_0^T (\Free[p_t^{\sched}] -
\Free[p_t^0])\,\dd t$, which weights distortion across the entire
horizon rather than only at the end. The argument of
Section~\ref{sec:why-these-two} shows that
$\dF^{\rm int} \approx (T/2)\,\dF$ for fixed schedule along the
$\Schd_R$ contour: the integrand grows roughly linearly in $t$,
so the schedule ranking is preserved. Numerical evidence on the
canonical calibration confirms the two variants pick optima that
agree to four decimal places in $(k^{\star}, \tw^{\star})$.

\paragraph{Bracket-family numerical comparison.} A
small-threshold expansion from the (C1)--(C3) limit
$X_0 \to 0$, computed numerically on a 2-bracket sub-class
$(k, X_0, r)$, confirms that the JKO criterion's preference
for the matched-revenue interior is robust to the introduction
of small thresholds. The leading-order correction is
$O(X_0)$ in the JKO value and reduces the bluntness penalty
$B(m_0)$ on the high-wealth tail, as predicted by
Section~\ref{sec:bluntness}. The full development is the
subject of the bracket-extension companion paper
(Section~\ref{sec:open-questions}); the present check
establishes only that the (C1)--(C3) results are not
artefacts of the proportional restriction --- they survive
small perturbations in the schedule structure.

Each check confirms the headline finding: the linear-vs-quadratic
bluntness contrast and the phase structure of the JKO optimum
are properties of the criterion choice, not of the particular
revenue linearisation, the terminal-vs-integrated functional
form, or the proportional-rate restriction.

\subsection{Implementability}\label{sec:implementability}

The continuous-time, perfect-assessment, proportional-rate
framework of this paper is an idealisation of the actual policy
instruments through which wealth taxes are administered.
Three implementation channels are worth flagging.

\paragraph{Bracket-based step functions and the proportional
idealisation.} Real wealth-tax schedules are piecewise-linear
in wealth: an exemption threshold and one or more marginal
rates above. The proportional-rate idealisation $\tw \in [0,
\tw^{\rm max}]$ used here is the simplest case of this
hierarchy and the cleanest setting for the closed-form
analysis. The bluntness mechanism of
Section~\ref{sec:bluntness} explains why bracket schedules
emerge as economically motivated even from the JKO
perspective; the bracket-extension companion paper sketched
in Section~\ref{sec:open-questions} treats the
implementability question rigorously.

\paragraph{Yearly assessment and \emph{forskuddsskatt}.}
Norwegian wealth-tax assessment is yearly and on the
self-reported balance-sheet at year-end, with prepayment
(\emph{forskuddsskatt}) collected throughout the year against
the previous year's assessment. The framework's continuous-time
formulation aggregates over this yearly granularity: the
horizon $T = 5$ used in the canonical calibration corresponds
to a five-year planning window over which yearly assessments
are summed, and the matched-revenue contour
$a\,\tw + b\,(1-k) = R^{\star}$ is the time-integrated revenue
target. A reader interested in within-year assessment dynamics
should read $\tw$ as the yearly-equivalent rate.

\paragraph{Continuous-time as long-horizon approximation.}
The Fokker--Planck framework treats wealth as a continuous-time
diffusion. At short horizons (sub-annual, say), the
Gaussian-on-log-wealth approximation is rough --- jumps,
rebalancing events, and discrete tax assessments produce
deviations from pure GBM dynamics. At long horizons (multi-year
planning, including the canonical $T = 5$), the GBM
approximation is empirically robust, and the closed-form FOCs
of Section~\ref{sec:closedform} apply with quantitative accuracy. The
$W_2$-vs-JKO contrast is itself a long-horizon phenomenon: it
sharpens with $T$ as the displacement penalty grows linearly
in $T$ under JKO and quadratically under $W_2$.

The implementability detail is therefore secondary to the
within-paper analysis: the framework's predictions are robust to
yearly-versus-continuous assessment and hold quantitatively in
the long-horizon regime that the canonical calibration targets.
The remaining implementability question --- the bracket
structure --- is sufficiently substantive to constitute a
separate paper.

\subsection{Connection to the Wasserstein
companion}\label{sec:p10-link}

\citet{Froeseth2026W} (henceforth the Wasserstein companion)
establishes that the $W_2$-optimal schedule on the full admissible
class $\Schd \supset \mathcal{N}$ is the continuously progressive
$\sched^{\star}(x) = \lambda\,x^2$, with rate
$r(x) = \lambda x$ rising linearly in wealth. That schedule
lies outside the (C1)--(C3) neutrality class: it violates (C3)
by construction (the assessment is wealth-dependent), and it
sacrifices neutrality for distributional compression in the
sense made precise in Section~\ref{sec:p1-link}. Within the
restricted class (C1)--(C3) of the present paper, the
$W_2$-optimal schedule is the corner-pinning at
$(k_{\rm lo}, 0)$ documented in Theorem~\ref{thm:w2opt} and
Remark~\ref{rem:w2-no-transition}: that is the
\emph{closest (C1)--(C3)-feasible point} to the unconstrained
quadratic $\sched^{\star}$, modulo the projection onto the
proportional sub-class.

The two papers therefore fill complementary cells of a $2\times 2$
grid: criterion (JKO or $W_2$) crossed with schedule class
(restricted (C1)--(C3) or full $\Schd$). The Wasserstein
companion fills the $W_2$-on-$\Schd$ cell; the present paper
fills both JKO-on-(C1)--(C3) and $W_2$-on-(C1)--(C3). The
JKO-on-$\Schd$ cell --- the JKO-optimal schedule on the
unrestricted class --- is the natural fourth cell and is flagged
as an open question in Section~\ref{sec:open-questions}.

The clean reading of the joint structure is: \citet{Froeseth2026W}
establishes that the $W_2$ criterion, given full schedule freedom,
picks a non-neutral progressive schedule. The present paper
establishes that within the neutrality-preserving class, the JKO
and $W_2$ criteria still disagree --- now along the
proportional-mix dimension rather than along the neutrality
dimension --- and that the disagreement is regime-dependent
through $\rho$. The two findings are not in tension: they
describe the criteria's preferences along different axes of the
schedule space. The companion paper's result is the answer to
``what schedule does $W_2$ pick when allowed?''; the present
paper's result is the answer to ``what schedule does each
criterion pick when restricted to neutrality?''. Both questions
are well-posed, and the relationship between the two is mediated
by the neutrality-vs-distributional-compression normative split
foregrounded in Section~\ref{sec:traditions}.

\subsection{Heterogeneous returns and other extensions}\label{sec:heterogeneous}

The framework treats wealth as a single GBM with homogeneous
parameters $(\mu, \sigma)$. Two extensions are immediate.

\paragraph{Multi-asset / heterogeneous-returns extension.} A
realistic Norwegian household holds a portfolio of asset classes
with distinct return distributions: equities, real estate, bonds,
deposits, business assets. The single-GBM idealisation aggregates
these into an effective $(\mu_{\rm eff}, \sigma_{\rm eff})$
along the lines of Section~\ref{sec:effective-sigma}. A
multi-asset extension would replace the single $X_t$ with a
vector $\bm{X}_t = (X_{1,t}, \ldots, X_{n,t})$ of per-class
holdings, each evolving as its own GBM, and the wealth-tax
schedule would become asset-class-specific:
$\sched_i$ for class $i$. The (C3) condition $\alpha_i = \alpha$
corresponds to uniform assessment across asset classes; its
violation in actual Norwegian policy (the $25\%$ valuation rule
for primary residences, the discounted assessment of unlisted
shares) are exactly per-class deviations from a common
$\alpha$. A forthcoming heterogeneous-returns companion to the
present paper develops the JKO/$W_2$ optimisation on the
multi-asset state space and characterises the optimal
asset-class-specific assessment rates. The single-asset result
of the present paper is the homogeneous-returns specialisation
of that companion analysis.

The bluntness mechanism of Section~\ref{sec:bluntness} is
sharpest in the heterogeneous-returns setting. The bluntness
index $B(m_0) = b/(a m_0)$ depends on the geometric mean return
$m_0$, which varies across asset classes (and across cohorts of
households with distinct asset compositions). A wealth tax that
applies uniformly across classes is therefore proportionally
more punitive on low-$m_0$ classes than on high-$m_0$ classes:
the bluntness penalty is regressive in the cross-section of
returns. This is exactly the asymmetry documented in
\citet{Froeseth2026H}, and it is the principal motivation for
asset-class-specific assessment in actual policy.

\paragraph{Spectral-side connection.} \citet{Froeseth2026X}
develops a spectral portfolio-theory framework that decomposes
the wealth process by eigenmodes of the covariance structure.
Under (C1)--(C3) the spectral structure is preserved up to the
drift-shift-and-rescale, so the present paper's results lift to
the spectral framework without modification. Once the (C1)--(C3)
restriction is dropped, the asset-class-specific assessment rates
of the heterogeneous-returns companion map naturally onto the
spectral basis, with each eigenmode receiving its own optimal
$\alpha_i$. Whether the JKO and $W_2$ criteria agree on the
spectral-side optimum is a third open question that the spectral
/ heterogeneous-returns extension will answer.

\subsection{Phase-diagram structure as a research
direction}\label{sec:phase-diagram-direction}

Section~\ref{sec:phase-transition} formulated the JKO crossover as a
phase transition with order parameter $\phi^{\star}$, control
parameter $\rho$, and free energy $\Delta F^{\star}$. The figures
collapse the picture onto a one-dimensional slice through the GBM
parameter space at fixed horizon $T$ and fixed initial spread
$v_0$. The full phase diagram lives in higher dimensions, and
several axes invite exploration.

The most immediate generalisation is the
$(\rho, R^{\star})$ \emph{phase diagram}. The two boundaries
$\rho_{\rm low}(R^{\star})$ and $\rho_{\rm high}(R^{\star})$
depend on the revenue target $R^{\star}$, and the mixed window
$[\rho_{\rm low}(R^{\star}), \rho_{\rm high}(R^{\star})]$ widens
or narrows as $R^{\star}$ varies. At small enough $R^{\star}$ the
flow-tax base $b\,(1-k)$ alone can deliver the target across the
full feasibility range and the mixed phase pinches to a point ---
a candidate \emph{tricritical-style} structure where the two
phase boundaries collide. The geometry of that collision and the
behaviour of $\Delta F^{\star}$ near it are open analytical
questions.

A second generalisation extends the dimensionality of the regime
parameter. Treating $T$ and $v_0$ as additional axes opens a
four-dimensional GBM--policy phase manifold $(\sigma, T, v_0,
R^{\star})$. The phase boundaries become hypersurfaces, and the
mixed-phase volume becomes a quantitative measure of how often
JKO recommends a strictly mixed policy. A scan of the canonical
Norwegian-flavoured neighbourhood would tell us how robust
Norway's location in the mixed phase is to perturbations of the
underlying parameters --- a natural sensitivity test for the
political reading.

A third direction is the $W_2$-vs-JKO contrast made phase-
diagrammatic. In the calibration of Figure~\ref{fig:rho-crossover}
$W_2$ is phase-poor: it pins to the pure flow-tax corner for every
$\rho$. Whether $W_2$ acquires phase structure on a different
slice of the parameter space, and whether there is an
intermediate criterion that interpolates between the JKO and
$W_2$ free-energy weightings, would let us trace how
phase-richness emerges as a function of how the criterion weights
mean-distortion against spread-distortion. The answer would
sharpen the political reading of the criterion choice.

What does \emph{not} transfer from the statistical-mechanics
analogy is, as noted in Section~\ref{sec:phase-transition}, divergent
susceptibility, critical exponents in the universality-class
sense, and hysteresis. The kinks in $\phi^{\star}(\rho)$ are
constraint-set crossings rather than spontaneous-symmetry-breaking
events, and the relevant universality class is the unremarkable
one of polytope-constrained convex optimisation. The point is
that the phase-transition language is genuinely informative ---
it organises the mixed-versus-corner regimes with the right
vocabulary --- without claiming a depth of physical analogy that
the underlying mathematics does not justify.

\subsection{Open questions}\label{sec:open-questions}

\paragraph{The 2-bracket extension and the companion paper.}
The principal limitation of the present analysis is its
restriction to proportional wealth taxes. Real wealth-tax
schedules in the jurisdictions where the instrument is
operational --- Norway, Switzerland, Spain, France's old ISF ---
are piecewise-linear in wealth: an exemption threshold $X_0$ and
one or two non-zero marginal rates above. The simplest
non-trivial case, the \emph{2-bracket sub-class}
$\{\sched_{X_0, r}(x) = r\,(x - X_0)_+ : X_0 \geq 0,\,
r \geq 0\}$, has three free parameters $(k, X_0, r)$ given the
flow-tax retention $k$, and is the natural next layer of the
schedule hierarchy
\[
\text{(C1)--(C3) (proportional, 1-bracket)}
\;\subset\;
\text{2-bracket}
\;\subset\;
\cdots
\;\subset\;
\Schd \quad \bigl[\text{the full class of \citet{Froeseth2026W}}\bigr].
\]
The companion paper has two complementary strands.

\emph{Semi-analytical strand.} Within each bracket, the FP
equation has constant-coefficient drift on log-wealth, the
post-tax density is locally Gaussian with a drift correction, and
the bracket boundary is a matching condition (continuity of
density and probability flux) that admits a closed-form solution
in terms of the error function. $\Delta F$ and $W_2^2$ on the
2-bracket sub-class therefore remain analytic --- in elementary
functions plus $\mathrm{erf}$ --- and the matched-revenue
Lagrangian becomes a 3-parameter optimisation with explicit
(transcendental) FOCs. The phase-transition mechanism of
Section~\ref{sec:mechanism} extends in spirit: the constant-
marginal-cost-in-$\tw$ structure of JKO generalises to a
piecewise-explicit marginal cost in $r$ with $X_0$-dependent
coefficients, and the boundaries of the JKO phase diagram in
the new $(\rho, R^{\star}, X_0)$ parameter space can be
characterised analytically near the (C1)--(C3) limit
$X_0 \to 0$.

\emph{Numerical strand.} The semi-analytical approach is clean
on the 2-bracket sub-class but does not extend straightforwardly
to richer schedules ($n$-bracket for $n \geq 3$, smoothly
progressive shapes, or schedules with non-trivial boundary
behaviour) where the post-tax density loses its piecewise-Gaussian
structure. A direct numerical attack --- solving the FP equation
by finite-difference time-stepping or Monte Carlo simulation of
the post-tax SDE, evaluating $\Delta F$ and $W_2^2$ from the
resulting density, and optimising over the schedule
parameters --- handles the entire schedule space and is the right
tool for mapping the full phase diagram across $(\mu, \sigma, T,
v_0, R^{\star}, X_0, r, \ldots)$ at once. The numerical strand
also serves as verification of the semi-analytical closed forms
where they apply, and as the only available tool in regimes where
they do not.

The natural framing is that of a two-paper series. The present
paper establishes the (C1)--(C3) phase structure on the
proportional sub-class, pinning down the regime-dependence and
the JKO-vs-$W_2$ contrast in the cleanest setting. The companion
paper combines the two strands above to deliver (i) closed-form
JKO-optimal $(X_0^{\star}, r^{\star})$ at the leading order in a
small-threshold expansion from $X_0 = 0$, (ii) numerical
solutions across the full 2-bracket parameter space and beyond,
mapping the JKO phase diagram and validating the analytical
results, (iii) the matching empirical comparison to the actual
Norwegian and Swiss bracket schedules, which the present paper
deliberately defers. The structural relationship between the
two papers parallels the neutrality / extensions split between
\citet{Froeseth2026N} and \citet{Froeseth2026F}.

\paragraph{Other open questions.}
Two further analytical questions remain, treated only briefly
here: sweeps over $(\mu, v_0)$ at fixed $\sigma$ to trace the
full four-dimensional crossover surface (this paper traces
only the $\sigma$-slice and the $T$-slice); and the
unconstrained mathematical optimum allowing schedule-side
negative coefficients as a curiosity case.

\subsection{Horizon dependence of the JKO recommendation}\label{sec:horizon}

The phase structure of Theorem~\ref{thm:crossover} is established at a
fixed planning horizon $T$, and the calibration of
Section~\ref{sec:norwegian} fixes $T = 5$ years as the canonical
Norwegian-baseline value. The structural question we address
here is what happens to the phase boundaries
$\rho_{\rm low}(T), \rho_{\rm high}(T)$ as the horizon $T$ is
varied, and where Norway's actual planning horizon falls in that
picture. The answer connects directly to the bluntness mechanism
of Section~\ref{sec:bluntness} via the natural scale crossover
$T_c = v_0/\sigma^2$ introduced at the end of that section.

\paragraph{Two scaling regimes of $\rho$.} The crossover ratio
$\rho = \Sigma_0\,m_0/\sigma^2$ has two distinct $T$-scalings,
separated by the time at which the cumulative diffusion
$\sigma^2 T$ matches the initial log-wealth variance $v_0$:
\begin{equation}\label{eq:rho-regimes}
\rho \;\approx\;
\begin{cases}
\sqrt{v_0}\,\dfrac{m_0}{\sigma^2}
& \text{for } T \ll T_c \quad \text{(initial-spread-dominated)}, \\[0.6em]
\dfrac{m_0\,\sqrt{T}}{\sigma}
& \text{for } T \gg T_c \quad \text{(diffusion-dominated)}.
\end{cases}
\end{equation}
For Norwegian-baseline parameters $(v_0 = 0.25,\,\sigma = 0.30)$,
$T_c \approx 2.78$ years; the canonical $T = 5$ sits just past
the crossover, in the diffusion-dominated regime where $\rho$
grows like $\sqrt{T}$.

\paragraph{Both phase boundaries pinch at $T = 0$ and $T \to
\infty$.} Substituting the unconstrained closed-form
$k^{\star}_{\rm JKO}(\sigma) = K$ into the matched-revenue
condition gives, after simplification,
\begin{equation}\label{eq:phase-boundary-condition}
\sigma^4 \;=\; 2\,m_0\,(b/a - m_0)\,(K\,\sigma^2 T + v_0).
\end{equation}
Two limits:

\emph{Short-horizon limit ($T \to 0$).} The $K\,\sigma^2 T$ term
vanishes and \eqref{eq:phase-boundary-condition} reduces to
$\sigma^4 = 2\,m_0\,(b/a - m_0)\,v_0$ --- \emph{independent of
$K$}. The two phase boundaries (which differ only in their target
$K \in \{1, k_{\rm lo}\}$) collapse onto the same $\sigma$, and
the mixed-phase $\sigma$-window vanishes.

\emph{Long-horizon limit ($T \to \infty$).} The $K\,\sigma^2 T$
term dominates and $\sigma^2 \to 2\,m_0\,(b/a - m_0)\,K\,T$. To
keep $\sigma^2$ bounded by the existence condition $\sigma^2 < 2\mu$,
$m_0$ must shrink like $1/T$, which forces $\sigma \to \sqrt{2\mu}$.
Both boundaries converge to the same limit. The $\rho$-width
decays as $\sim 1/\sqrt{T}$ but never quite reaches zero in
finite $T$.

The two pinches have different content. The $T \to 0$ pinch is
``no time for the JKO criterion to discriminate between the two
channels''. The $T \to \infty$ pinch is ``the geometric mean
log-return $m_0$ vanishes, both channels look equivalent at the
margin''. Between the two, the mixed-phase window opens, peaks
in width at $T^{\star} \approx 11$ years, then closes again.

\paragraph{Norway's planning horizon flips the recommendation.}
Norway's $\sigma = 0.30$ trajectory through the
$(T, \rho)$ plane is shown in Figure~\ref{fig:horizon-phase-diagram}.
Because $\rho_{\rm Norway}(T)$ grows monotonically from
$\sqrt{v_0}\,m_0/\sigma^2 \approx 0.14$ at $T \to 0$ toward
the diffusion-dominated branch as $T$ grows, while the phase
boundaries $\rho_{\rm low}(T), \rho_{\rm high}(T)$ shift in the
opposite direction, the trajectory crosses both boundaries.
Numerically:

\begin{center}
\begin{tabular}{lll}
\toprule
$T$ (years) & $\rho_{\rm Norway}$ & JKO recommendation \\
\midrule
$1$  & $0.16$ & Pure wealth-tax phase \\
$3$  & $0.20$ & Pure wealth-tax phase \\
$\mathbf{5}$ \emph{(canonical)} & $\mathbf{0.23}$ & \textbf{Mixed-instrument phase} \\
$10$ & $0.30$ & Pure flow-tax phase \\
$20$ & $0.40$ & Pure flow-tax phase \\
\bottomrule
\end{tabular}
\end{center}

The mixed-instrument recommendation that the present paper
foregrounds at the canonical Norwegian-baseline regime is
\emph{specific to the five-year planning horizon}. Shorter
horizons (one-year tax-planning windows, year-to-year revenue
adjustments) push the JKO optimum into the pure wealth-tax phase;
longer horizons (retirement, intergenerational planning) push it
into the pure flow-tax phase. The qualitative recommendation
flips twice as $T$ varies.

\paragraph{Economic content of $T_c$.} The scale crossover
$T_c = v_0/\sigma^2$ has a direct interpretation: it is the
time at which the wealth distribution stops ``remembering'' its
initial state and starts being shaped predominantly by the
post-tax dynamics. Below $T_c$, the cost structure under JKO is
dominated by initial-condition terms and does not discriminate
between the two channels; above $T_c$, the channels' time-scaling
diverge and the bluntness penalty's $T$-amplification under the
wealth-tax channel makes that channel progressively more
expensive.

The horizon-dependent flip is therefore not a quirk of the model
but a reflection of which time-scale the policy debate is
implicitly using. Year-to-year revenue planning lives in a
fundamentally different regime from lifecycle wealth-tax design,
and the JKO criterion --- which prices distortion linearly in
the displacement and accumulates that displacement linearly in
$T$ --- registers the difference. A reader uncomfortable with
the canonical mixed-phase recommendation has, in this picture, a
principled out: choose a different planning horizon, and the JKO
criterion will recommend a different schedule. A reader committed
to the mixed-phase recommendation has, equally, a principled
defence: only at the five-year horizon does the criterion give a
strictly mixed answer, and that horizon corresponds to a natural
equity-portfolio planning window.

\begin{figure}[!htbp]
\centering
\IfFileExists{figures/fig_horizon_phase_diagram.pdf}{%
  \includegraphics[width=0.95\textwidth]{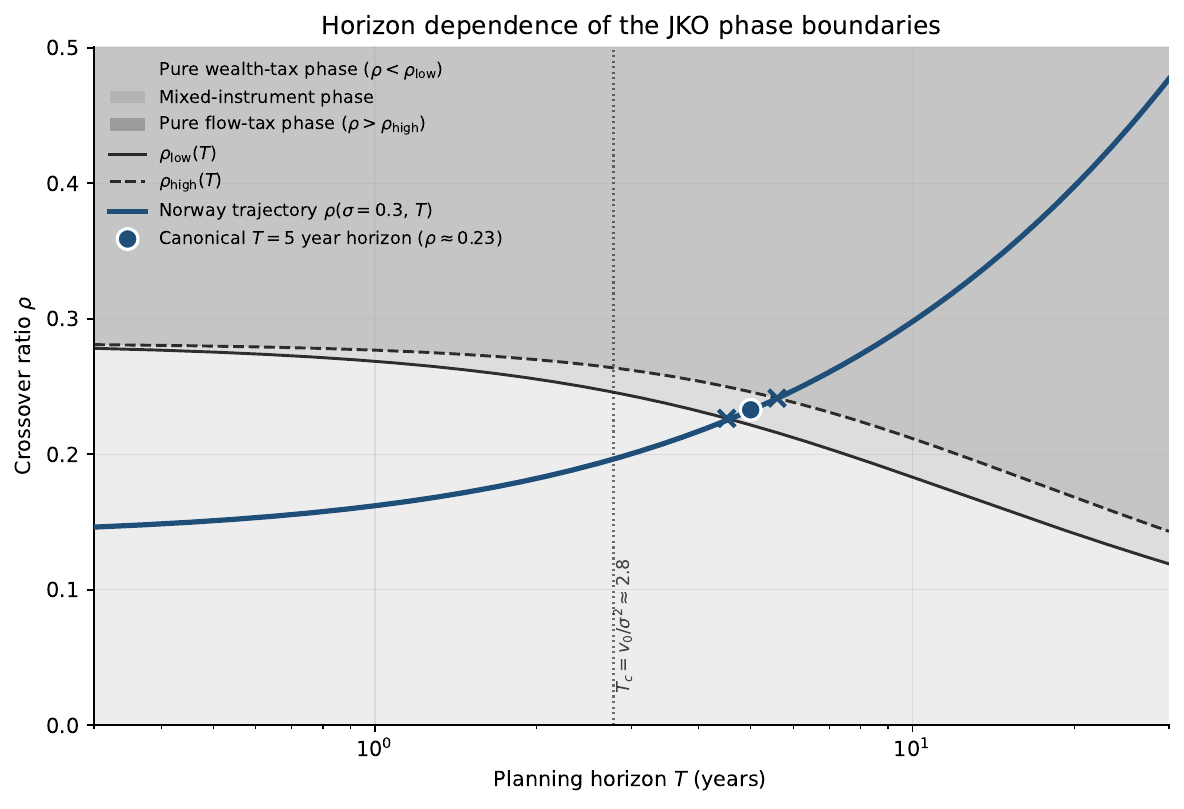}%
}{%
  \fbox{\parbox{0.95\textwidth}{\centering\vspace{1.5em}%
  \textbf{Figure placeholder.}\\[0.4em]
  Run \texttt{fp\_design/figures/fig\_horizon\_phase\_diagram.py}
  locally to generate \texttt{fig\_horizon\_phase\_diagram.pdf}.\vspace{1.5em}}}%
}
\caption[Horizon dependence of the JKO phase boundaries with
the Norway-$\sigma$ trajectory overlaid]{%
\textbf{Horizon dependence of the JKO phase boundaries with the
Norway-$\sigma$ trajectory overlaid.}
The phase boundaries $\rho_{\rm low}(T)$ (solid grey) and
$\rho_{\rm high}(T)$ (dashed grey) define the boundaries of the
mixed-instrument phase as functions of the planning horizon
$T$ (log-scaled). The three shaded regions are the JKO phases
of Theorem~\ref{thm:crossover}: gold (pure wealth-tax,
$\rho < \rho_{\rm low}$), lavender (mixed instrument,
$\rho_{\rm low} \leq \rho \leq \rho_{\rm high}$), cyan (pure
flow-tax, $\rho > \rho_{\rm high}$). The teal curve is
$\rho(\sigma = 0.30,\,T)$, the Norway-$\sigma$ trajectory through
the $(T, \rho)$ plane. The two crosses mark where the
trajectory enters and leaves the mixed-instrument phase. The
canonical Norwegian-baseline $T = 5$ years (teal disc) sits
inside the mixed phase by a narrow margin; shorter horizons fall
in the pure wealth-tax phase, longer horizons in the pure
flow-tax phase. The dotted vertical rule at
$T_c = v_0/\sigma^2 \approx 2.78$ marks the scale crossover from
the initial-spread-dominated short-horizon regime to the
diffusion-dominated long-horizon regime. Calibration: $\mu = 0.07$,
$\sigma = 0.30$, $v_0 = 0.25$, $a = 1.0$, $b = 0.27$,
$R^{\star} = 0.05$, matched to Figure~\ref{fig:rho-crossover}.}
\label{fig:horizon-phase-diagram}
\end{figure}

\section*{Acknowledgements}

The author acknowledges the use of Claude (Anthropic) for assistance
with literature review, \LaTeX{} typesetting, mathematical exposition,
and editorial refinement, and Lemma (Axiomatic AI) for review and
proof checking. All substantive arguments, economic reasoning, and
conclusions are the author's own.

\end{document}